\documentclass[fullpage]{article}

\usepackage{fullpage}
\usepackage{color}
\usepackage{enumitem}
\usepackage{authblk}
\usepackage {xspace}
\usepackage{amsmath}
\usepackage{amssymb}
\usepackage{tabu}
\usepackage{framed}
\usepackage{mathtools}
\usepackage{multirow}
\usepackage{lipsum}

\usepackage{algorithmic}
\usepackage{algorithm}
\usepackage{listings}
\newcommand{\cc}{\mathrm{C}}
\newcommand{\Alg}{\mbox{ALG}\xspace}
\newcommand{\LIS}{\mbox{LIS}\xspace}
\newcommand{\cC}{{\mathcal C}}
\newcommand{\cP}{{\mathcal P}}
\renewcommand{\ell}{\pi}
\newcommand{\lmin}{{\ell_{min}}\xspace}
\newcommand{\lmax}{{\ell_{max}}\xspace}

\newcommand{\LISs}{m\mbox{-LIS}\xspace}

\newcommand{\SLBurst}{\rho\mbox{m-Preamble}\xspace}

\newcommand{\GroupLIS}{\mbox{GroupLIS}}
\newcommand{\Greedy}{\mbox{Greedy}\xspace}
\newcommand{\MGreedy}{\mbox{MGreedy}\xspace}
\newcommand{\Preamble}{\mbox{SL-Preamble}\xspace}

\newcommand{\rhoflr}{\overline{\rho}}
\newcommand{\Amortized}{k\mbox{-Amortized}\xspace}

\newtheorem{definition}{Definition}
\newtheorem{theorem}{Theorem}
\newtheorem{lemma}{Lemma}
  
\newtheorem{observation}{Observation}  
\newtheorem{claim}{Claim}[section]

\newcommand{\sq}{\hbox{\rlap{$\sqcap$}$\sqcup$}}
\newcommand{\qed}{\hspace*{\fill}\sq}

\newenvironment{proofof}[1]{\bigskip \noindent {\bf Proof of #1:}}

\newenvironment{proof}{\noindent{\bf Proof:}}{\hfill\rule{2mm}{2mm}\\}    

\newcommand{\ez}[1]{{\color{magenta} #1}}

\newcommand{\remove}[1]{}

\begin{document}

\title{
Online Distributed Scheduling on a Fault-prone Parallel System
}




\author[1,2]{Elli Zavou \thanks{elli.zavou@imdea.org}}
\author[1]{Antonio Fern\'andez Anta}
\affil[1]{IMDEA Networks Institute, 28918, Legan\'es (Madrid), Spain}
\affil[2]{Universidad Carlos III de Mdrid, Madrid, Spain}

%


\date{}


\maketitle 

\begin{abstract}
We consider a parallel system of $m$ identical machines prone to unpredictable crashes and restarts, trying to cope with the continuous arrival of tasks to
be executed. Tasks have different computational requirements (i.e., processing time or \emph{size}). The flow of tasks, their size, and the crash and restart of
the machines are assumed to be controlled by an adversary.
Then, we focus on the study of online distributed algorithms for the efficient scheduling of the tasks.
We use competitive analysis, considering as efficiency metric the {\em completed-load}, i.e., the aggregated size of the completed tasks.

We first present optimal completed-load competitiveness algorithms when the number of different task sizes that can be injected by the adversary is
bounded. (It is known that, if it is not bounded, competitiveness is not achievable.) We first consider only two different task sizes, 
and then proceed to $k$ different ones, showing in both cases that the optimal completed-load competitiveness can be achieved.

Then, we consider the possibility of having some form of resource augmentation, allowing the scheduling algorithm to run with a speedup $s \geq 1$.
In this case, we show that the competitiveness of all work-conserving scheduling algorithms can be increased by using a large enough speedup.
\end{abstract}


%


\section{Introduction}
\label{sec:intro}

With the widespread use of cloud computing (which is essentially equivalent to computing in large scale data centers) and big data processing, parallel computing
is taking new forms. Parallelism appears as the execution of lightly coupled tasks (or jobs), like the map or reduce tasks of a map-reduce computation, in a collection of
decoupled processors (or cores). The large scale of both size and time of these types of computation, has two important consequences. First, 
it makes it highly likely that processors will fail during the computation (failures are the norm, not the exception \cite{Ekanayake2010,Zhang2010,jhawar2013fault,bala2012fault}), and hence recovery mechanisms must be an
intrinsic part of the task scheduler. Second, it is unlikely that the information about all the tasks to be scheduled is available at the initial time of the computation, which means that
the task scheduler must make online decisions. These two aspects make most prior work on task scheduling on parallel 
machines not applicable in these new environments (e.g.,~\cite{el1995task,gharbi2005optimal,166609,KS97,
sanlaville1998machine,shirazi1990analysis}).
For instance, some works tackle the issue of dynamic job arrivals but do not consider failures (e.g.,~\cite{AKP_STOC92}), 
others consider failures but assume knowledge of all the jobs a priori (e.g.,~\cite{Faith2013,kalyanasundaram1994fault}), finally, some others
consider energy efficiency issues but not machine failures, or only jobs of the same computational demand (e.g.,~\cite{chan2009optimizing}).

\begin{figure}[h!]
  \centering
    \includegraphics[width=0.45\textwidth]{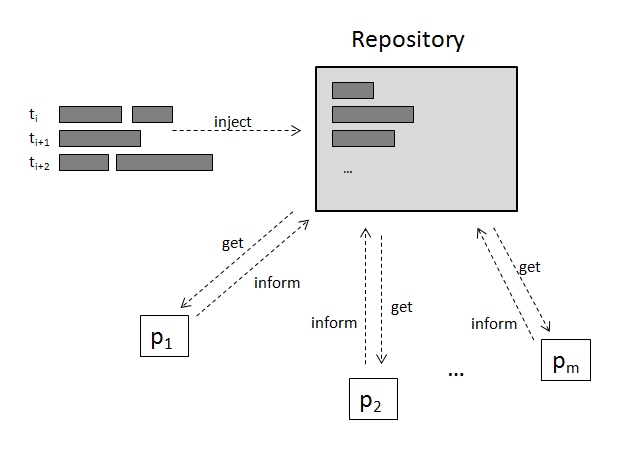}
    \caption{\small The computing setting considered, with the $m$ homogeneous machines, the shared Repository, and the three operations; {\em inject}, for the dynamic task arrivals from the users, {\em get}, for the machines to obtain the set of pending tasks, and {\em inform} for the repository to update the set of pending tasks.}
    \label{fig:diagram}
\end{figure}

In this paper, we explore the scheduling of tasks in a parallel system like the one depicted in Figure~\ref{fig:diagram} (which is similar to the one
introduced in \cite{anta2015online}). The system has
$m$ identical machines\footnote{We use the terms \emph{machine} and \emph{processor} interchangeably.}
 prone to crashes and restarts (controlled by an adversary). Independent idempotent tasks with different computational demands (e.g., in terms of processing time) arrive into 
the system, to be executed by any of the machines. In order to remove a single point of failure, the system considered has no central scheduler. Instead (see Fig.~\ref{fig:diagram}),
each task that arrives is held in a repository until some machine executes it and reports this fact. Hence, the scheduling is done in a distributed way
by the machines. The objective is to design efficient distributed scheduling algorithms that run at the machines.
In summary, the characteristics of the parallel system considered here are (1) continuous task arrivals,
(2) machine failures and restarts, and (3) distributed scheduling.
We measure the performance of an algorithm by its completed-load (i.e., the total size of the completed tasks) competitiveness.

One additional aspect of cloud computing is its impact in terms of energy consumption. It is known that the power consumed by a processor grows with
its processing speed (in cycles per second)~\cite{Arjona2014}. In our model we introduce this in the form of a speedup $s \geq 1$,
that can be used to allow processors to run faster than the baseline. This is a form of resource augmentation. We show that this
resource augmentation, increasing the speed of the machines, improves the
competitiveness of the system.

\paragraph{Related Work}
Probably, the most important research line related to this work is the study of machine scheduling with availability constraints 
(e.g.,~\cite{sanlaville1998machine,gharbi2005optimal,
georgiou2015competitiveness,Faith2013}). One of the most important outcomes of this line is the necessity of algorithms that take into account unexpected machine breakdowns. Most works allow preemptive scheduling 
\cite{gharbi2005optimal,166609} and show optimality only for nearly online algorithms (for example, algorithms that need to know the time of the next job arrival or machine availability).
Among these works some also consider energy issues, and use speed-scaling to tune the power consumption of the processors (e.g.,~\cite{Albers:2011:MSS:1989493.1989539,Chan:2009:SSP:1583991.1583994,chan2009optimizing}).


The work of Georgiou and Kowalski~\cite{georgiou2015competitiveness} was the one that initiated our study. They looked at a rather different setting, consisting of a cooperative computing system of $m$ message-passing processors prone to crashes and restarts, having to collaborate in order to complete dynamically injected tasks. For the efficiency of the system, they performed competitive analysis, focusing on the maximum number of pending tasks. They proved competitiveness with unit-length tasks, and showed that if tasks have different lengths competitiveness cannot be achieved.

In~\cite{anta2015online} we looked at a setting similar to the one used here, with $m$ machines, introduced the term of speedup, representing the resource augmentation required, and showed competitiveness in terms of pending load (sum of sizes of pending tasks).
More precisely, tasks of at least two different sizes were considered, and it was found that the threshold for competitiveness was defined by these two conditions: (a) $s < \rho$ and (b) $s < 1 + \gamma/\rho$, where $\rho$ is the ratio between the largest and smallest task and $\gamma$ is a parameter that depends on $\rho$ and $s$ (see~\cite{anta2015online} for details). If both conditions hold then \emph{no} deterministic algorithm is pending-load competitive. Then, some algorithms were proposed that achieve competitiveness as soon as one of the conditions does not hold; unfortunately imposing other restrictions, like considering only two different task costs.
A follow-up work~\cite{Anta2015competitive} compared popular scheduling algorithms (like FIFO, LIFO or LIS -- Longest in System) 
on the basic model of one machine, and looked at different efficiency measures, including the completed-load and latency competitiveness.
%
Kowalski et al.~\cite{kowalski2015fault}, also inspired by~\cite{anta2015online}, proved that in a system with one machine, and for speedup satisfying the conditions (a) and (b) as described above, no deterministic algorithm can be $1$-completed-load competitive.
They then proposed an algorithm that achieves 
$1$-completed-load competitiveness
if $s \geq 1+\gamma/\rho$.

In~\cite{Anta2013measuring}, a different setting was considered: an unreliable communication link between two nodes.
The problem of scheduling packets over such a link, is very related to the problem in this work, closely resembling the problem of 
scheduling tasks in one single machine with crashes and restarts (and without speedup).
In~\cite{Anta2013measuring}, the authors proposed the metric of {\em asymptotic throughput} for the performance evaluation of scheduling algorithms studied, which corresponds to the long term completed-load in our current setting.
Assuming only two packets lengths, they showed that for adversarial arrivals there is a tight value for the asymptotic throughput, giving upper bound with a fixed adversarial strategy and a matching lower bound with an online algorithm.
%
Jurdzinski et al.~\cite{Jurdzinski2015} extended that work, presenting optimal online algorithms for the case of $k$ fixed packet lengths, 
matching the upper bounds on the asymptotic throughput shown in~\cite{Anta2013measuring}.
Finally, they sketch a modification to one of the algorithms, in order to adapt it for the case of $f$ independent channels (stating that the analysis is not trivial).
This modified algorithm cannot be used in our setup because it uses
central scheduling; the sender has updated information and full control of the channel through which each packet is transmitted.

\paragraph{Contributions}

\begin{table*}[!t]\small
\begin{center}
\setlength{\tabcolsep}{1.5mm}
{\tabulinesep=1.5mm
\begin{tabu}{|c||c|c|}	
\hline
{\bf Speedup} & {\bf $\mathbf{1}$ - machine} & {\bf $\mathbf{m}$ - machines} \\
\hline \hline
 & \multicolumn{2}{c|}{$\cC(\Alg_W) = 0$, {\bf any} task size, \hfill~\cite{Anta2015competitive}} \\
 \cline{2-3}
\multirow{7}{*}{ $s = 1$ } & \multicolumn{2}{c|}{$\cC(\Alg_W) \leq \frac{\bar{\rho}}{\bar{\rho} + \rho} \approx 1/2$, {\bf two} task sizes, \hfill~\cite{Anta2013measuring}} \\
  \cline{2-3}
 & $\cC(\Preamble) \geq \frac{\bar{\rho}}{\bar{\rho} + \rho} \approx 1/2$, \hfill~\cite{Anta2013measuring} &  $\cC(\SLBurst) \geq \frac{\bar{\rho}}{\bar{\rho} + \rho} \approx 1/2$, \hfill [Th.~\ref{t:rmpreamble}] \\
 & {\bf two} task sizes & {\bf two} task sizes \\
 \cline{2-3}
 & $\cC(\Greedy) \geq \frac{\bar{\rho}}{\bar{\rho} + \rho} \approx 1/2$, \hfill~\cite{Jurdzinski2015} & $\cC(\Amortized) \gtrsim 1/2$, \hfill [Th.~\ref{th:Amortized}] \\
 & $k$ task sizes, {\bf pairwise divisible} &  $k$ task sizes, {\bf pairwise divisible} \\
 \cline{2-3}
 & $\cC(\MGreedy) \geq \min\limits_{1\leq j < i \leq k} \Big\{ \frac{\overline{\rho_{i,j}}}{\overline{\rho_{i,j}} + \rho_{i,j}} \Big\}$, \hfill~\cite{Jurdzinski2015} & $\cC(M\Amortized) \geq \min\limits_{1\leq j < i \leq k} \Big\{ \frac{\overline{\rho_{i,j}}}{\overline{\rho_{i,j}} + \rho_{i,j}} \Big\}$, \hfill [Th.~\ref{th:mkAmortized}] \\
 & $k$ task sizes $\ell_1,\ldots,\ell_k$; $\rho_{i,j}=\ell_i/\ell_j$ & $k$ task sizes $\ell_1,\ldots,\ell_k$; $\rho_{i,j}=\ell_i/\ell_j$ \\
 & {\bf general} & {\bf general} \\
\hline \hline
 $s \geq \rho$ & $\cC(\Alg_W) \geq 1/\rho$, \hfill~\cite{Anta2015competitive} & $\cC(\Alg_W) \geq 1/\rho$, \hfill [Th.~\ref{t:srho}] \\
\hline \hline
$s \geq \max\{\rho, 2\}$ & $\cC(\LIS) \geq 1$ \hfill~\cite{Anta2015competitive}& $\cC(\LISs) < 1$, {\bf when} $m=s=\rho=2$ \hfill [Th.~\ref{t:mLIS}] \\
\hline \hline
$s \geq 1 + \rho$ & $\cC(\Alg_W) \geq 1$, \hfill~\cite{Anta2015competitive} & $\cC(\Alg_W) \geq 1$, \hfill [Th.~\ref{t:srhoplus}] \\
 \hline
\end{tabu}}
\caption{\small 
{\em Negative} (upper bounds) and {\em positive} (lower bounds) results on the completed-load competitiveness.
$\Alg_W$ is any work conserving algorithm.
Note that the negative results hold for both $1$-machine and $m$-machines. Recall that $\rho=\lmin/\lmax$ and $\bar{\rho} = \lfloor \rho \rfloor$.}
\label{table-general}
\end{center}
\end{table*}

As mentioned, in this work, we consider a setting with $m$ {\em machines} prone to crashes and restarts, controlled by an {\em adversary} (to model worst-case scenarios), and a {\em shared repository} (an entity that provides the service by which the clients of the system submit the tasks to be executed and notifies them of their completion -- see Fig.~\ref{fig:diagram}). Note that the shared repository is not a scheduler, since it does not make any decisions on the execution of the tasks. It is basically 
a passive storage that behaves as an interface between the clients that generate tasks and the machines. It also allows the machines to maintain information about the tasks 
that have not been executed yet.

Tasks arrive in the system continuously and have different computational demands (which is their {\em size}). We assume that each task $\tau$ has size $\ell(\tau) \in[\lmin,\lmax]$, where $\lmin$ and $\lmax$ correspond to the smallest and largest possible values respectively, and that $\ell(\tau)$ only becomes known at the moment of the task arrival. Tasks are held in the shared repository, which is later accessed by the machines in order to decide which task to execute next. Note, then, that the machines' decisions are taken in a distributed manner and without any communication between them. When a task is completed, the corresponding machine {\em informs} the shared repository (which in turn notifies the client). 

As mentioned, we consider the possibility of having \emph{resource augmentation} in the form of \emph{speedup} $s \geq1$~\cite{kalyanasundaram2000speed,AGM-ICALP11} (i.e., we increase the computational speed of the machines) in order to cope with the performance cutback from the machine failures (crashes) and restarts, as well as the lack of information for the future task arrivals. More precisely, we consider uniform speedup $s\geq 1$ for all the machines, under which a task $\tau$ is executed $s$ times faster, i.e., in time $\ell(\tau)/s$.

Since the scheduling decisions must be made in a continuous manner and without future knowledge (neither of the task arrivals nor of the machine crashes and restarts), 
we study the problem as an {\em online} scheduling problem~\cite{pruhs2004online,anta2015online,Anta2013measuring}.
Hence, for the evaluation of the different algorithms proposed, we 
use {\em competitive analysis}~\cite{ST_CACM85}, measuring the {\em total completed load} of the system: the sum of sizes of the completed tasks. More precisely, an algorithm is considered $\alpha$-completed-load competitive, also expressed as $\cC(\Alg) = \alpha$, with speedup $s$ if under any adversarial behavior its completed-load complexity is at least $\alpha$ times the completed-load complexity of any algorithm $X$, running with no speedup and the same adversarial behavior.
Fully detailed specifications of the model used are given in Section~\ref{sec:model}.


In Table~\ref{table-general}, we summarize our results, including also some relevant results found in previous works. Note that the upper bounds found for the case of one machine, hold directly for the case of $m$ machines, since the adversary can simply crash all machines except one, and force the corresponding adversarial scenarios to occur. However, the {\em positive} results (lower bounds) do not necessarily transfer from $1$ to $m$ machines.

The upper part of Table~\ref{table-general} presents the results obtained in Section~\ref{sec:noSpeedup} when the machines run without speedup (i.e. $s=1$).
As can be seen, for this case we present algorithms that achieve optimal completed-load competitiveness, in all the cases. 
Observe that, in~\cite{Anta2015competitive}, it was shown for one machine that no work-conserving algorithm\footnote{An algorithm is work conserving if it does not allow a machine to be idle if there are tasks to be executed in the repository.} 
can achieve competitiveness
if there is an arbitrary number of different task sizes. This, being a negative result, also holds for the case of $m$-machines. 
Then, we give three work-conserving algorithms, focusing on the cases of two task sizes and bounded number of task sizes, {\em with} and {\em without pairwise divisibility}; a property that holds between each pair of task sizes (explained further in Section~\ref{sec:noSpeedup}). For the three cases, algorithms $\SLBurst$, $\Amortized$ and M$\Amortized$ respectively, achieve optimal competitiveness, matching the upper bound shown in~\cite{Anta2013measuring}. These algorithms
are non-trivial generalization of algorithms proposed in \cite{Anta2013measuring} and \cite{Jurdzinski2015}.

The lower part of Table~\ref{table-general} presents the results obtained in Section~\ref{sec:Speedup} for systems running with speedup $s>1$. The first
interesting observation from these results is that, contrary to intuition, to move from one machine to multiple machines it is not enough to complement an
algorithm that works for $m=1$ with a mechanism that prevents redundant execution of tasks when $m>1$. 
This is shown in the case of $s \geq \max\{\rho, 2\}$ with the algorithm \LISs, proposed in~\cite{anta2015online},
which is the natural adaptation of \LIS to multiple machines. As observed, while \LIS guarantees 1-completed-load competitiveness in one machine,
\LISs cannot achieve that level of competitiveness even with 2 machines.
Fortunately, as shown, we have been able to generalize two important general positive results obtained for work-conserving algorithms in one machine to
multiple machines.


\section{Model and Definitions}
\label{sec:model}

The parallel system considered has $m$ identical machines, prone to crashes and restarts, with unique ids in the set $\{0,1,2,\ldots,m-1\}$. Please see also Fig.~\ref{fig:diagram} for the graphical representation of the system. 
As mentioned, the machines have access to a shared {\em repository}.
The repository supports three essential operations: {\em inject,~get,} and {\em inform}. The {\em inject} operation is executed by a client of the system to add a new task to the current set of tasks to be executed. We assume here that this operation is controlled by an adversary (as will be further discussed below). Operations {\em get} and {\em inform} are executed by the machines. A machine uses the {\em get} operation to obtain the set of {\em pending tasks}, i.e., the tasks that are in the repository because they were injected and no machine has notified their completion yet. For simplicity, we assume that the {\em get} operation is blocking, i.e., if the repository is empty when
executed, it waits until some new task is before returning (the set of newly injected tasks).
A machine then executes an {\em inform} operation when it has completed the task scheduled, notifying the repository about its completion. 
Then the repository removes immediately this task. We assume that the execution of these operations is instantaneous (takes negligible time),
except a {\em get} operation that blocks.

We consider machines running in {\em processing cycles}, controlled by the scheduling algorithm considered. Each cycle, starts with a {\em get} operation, a task execution, and an {\em inform} operation (if the task is completed). Since the repository operations ({\em get} and {\em inform}) are instantaneous, a processing cycle lasts the time needed for the scheduled task to be completed. 
We assume that machines run with a {\em speedup} $s \geq 1$ ($s=0$ means no speedup).
Then, a processing cycle lasts a time equal to the size of the task divided by the speedup $s$.
When a machine crashes, the cycle is interrupted and the progress in the task execution is lost.
If the machine later recovers, it starts a new cycle.


\paragraph{Event ordering} 
Since the injection of tasks by clients, the {\em get} operations, and the notification via {\em inform} operations of task completion by the machines 
may occur simultaneously, we define the following order among these events. We assume that in an instant $t$ the {\em inform} operations occur first,
then the injections, and finally the {\em get} operations. Hence, a {\em get} operation executed at time $t$ will include the tasks injected at time $t$ but not the ones completed at that time.


\paragraph{Tasks}
As already explained, computational tasks are injected to the system by the clients, with the {\em inject} operation at the repository. We assume that this operation is controlled by an arrival pattern $A$ (a sequence of task injections) defined by an omniscient adversary.
Each task 
$\tau$ has an {\em arrival time} $a(\tau)$ 
and a {\em size} $\ell(\tau)$, which is the time required to complete the task without speedup. The task attributes are only known at the time of its injection.
We use the term $\ell$-task to refer to a task of size $\ell \in [\lmin,\lmax]$. The values $\lmin$ and $\lmax$ are the smallest and largest possible
task sizes, and are usually assumed to be known by the scheduling algorithm. 
We also define parameter $\rho = \frac{\lmax}{\lmin}$ to be the ratio between the largest and smallest task sizes.

We assume that tasks are {\em atomic} in the sense that not executing one completely due to a crash implies that it has to be executed again from the start.
On the other hand, we assume the tasks to be {\em idempotent}~\cite{GS_book08}, which means that executing the same task more than once has the same effect as executing it only once.

\paragraph{Machine Crashes and Restarts}
For the machine crashes and restart, we consider an omniscient adversary, which is the same entity responsible for the task injections at the repository described above.
This means that the adversary is expected to coordinate injections, crashes, and recoveries.
In an execution of the system, the adversary defines an {\em error pattern} $E$, which is a list of crash and restart events, each associated with the time it occurs (e.g., $crash(t,p)$ is the event that that machine $p$ is crashed at time $t$).
We consider a machine $p$ being {\em alive} in time interval $I = [t,t']$, if it is operational at time $t$ and does not crash at any time $t'' \leq t'$. 

\begin{definition}
An {\em adversarial pattern} is a combination of arrival and error patterns $A$ and $E$. 
It is {\em admissible} only when at all time instants there is at least one machine alive. 
In our work we only consider admissible adversarial patterns.
\end{definition}

\paragraph{Notation}
We consider it useful to provide all extensively used notation.
Since it is essential to keep track of injected, completed and pending tasks at each time instant in an execution, we introduce sets $I_t(A)$, $N^s_t(X,A,E)$ and $Q^s_t(X,A,E)$, where $X$ is an algorithm, $A$ and $E$ the arrival and error patterns respectively, $t$ the time instant under consideration and $s$ the speedup of the machines.
$I_t(A)$ represents the set of injected tasks within the interval $[0,t]$, $N^s_t(X,A,E)$ the set of completed tasks within $[0,t]$ and $Q^s_t(X,A,E)$ the set of pending tasks at time instant $t$.
As implied by the event ordering defined above, $Q^s_t(X,A,E)$ contains the tasks that were injected by time $t$ inclusively, but not the ones completed before and up to time $t$. Observe that $I_t(A) = N^s_t(X,A,E) \cup Q^s_t(X,A,E)$ and that set $I$ depends only on the arrival pattern $A$, while sets $N$ and $Q$ also depend on the error pattern $E$, the algorithm run by the scheduler, $X$, and the speedup of the machine, $s$.
For simplicity, we omit the superscript $s$ in further sections of the paper. However, the appropriate speedup in each case is clearly stated at all times.

\sloppy{We use $L_\ell$ to refer to the subset of $Q^s_t(X,A,E)$ that includes only pending tasks of size $\ell$, and we assume an ascending order of tasks in each queue $L_\ell$, according to their arrival time. To simplify the presentation of the algorithms, in a list of pending tasks we number them starting with $0$. Then, for instance, the tasks in $L_\ell$ are numbered from $0$ to $|L_\ell|$.
What is more, we will use notation $|L_\ell(X,t)|$ to refer to the number of $\ell$-tasks pending in the execution of $X$ at time $t$.
In a similar way, we use notation $|N_p(X,t)|$ to denote the number of completed tasks by machine $p$ in the execution of algorithm $X$ at time $t$.
}

Finally, we include some definition that originally appeared in~\cite{anta2015online} and will be used in the rest of the paper.

\begin{definition}[\cite{anta2015online}]
An algorithm is of type {\bf\em \GroupLIS}, if all the following hold:
\begin{itemize}[leftmargin=3mm]
\item It separates the pending tasks into classes containing tasks of the same size.
\item It sorts the tasks in each class in increasing order with respect to their arrival time.
\item If a class contains at least $m^2$ pending tasks and a machine $p$ schedules a task from that class, then it schedules the $(p\cdot m)$th task in the class.
\end{itemize} 
\end{definition}

\paragraph{Efficiency Measures}
We evaluate our algorithms with the \emph{completed load} measure. Given an algorithm $\Alg$ running with speedup $s\geq 1$, and adversarial arrival and error patterns $A$ and $E$ respectively, we look at time $t\geq 0$ of the execution and focus on the completed load complexity. This means, that we look at the sum of sizes of the completed tasks up to time instant $t$:
$$
C^s_t(\Alg,A,E) = \sum\limits_{\tau \in N^s_t(\Alg,A,E)} \ell(\tau)
$$
Finding the algorithm, in other words computing the schedule, that maximizes the measure offline (having the knowledge of patterns $A$ and $E$) is an NP-hard problem~\cite{anta2015online,Anta2013measuring}.

We will also be using a slightly changed notation $C^X(t,\ell)$ -- resp., $P^X(t,\ell)$ -- to denote the completed load -- resp., pending load -- of size $\ell$ at time instant $t$ in the execution of algorithm $X$. Similarly, $C^X(t,<\ell)$ -- resp., $P^X(t,<\ell)$ -- refers to the completed load -- resp., pending load -- of size smaller than $\ell$ at time instant $t$ in the execution of $X$.

As already mentioned, since the system is dynamic in respect to the task arrivals and machine crashes and restarts, we view the scheduling problem of task as an online one, and pursue competitive analysis. Specifically, considering any time $t$ of an execution, any combination of arrival and error patterns, $A$ and $E$, and any algorithm $X$ designed to solve the scheduling problem, the completed load competitiveness of an algorithm $\Alg$ that runs with speedup $s \geq 1$ is defined as follows:\\
Algorithm $\Alg$ is $\alpha$-completed-load competitive if $\forall t,X,A,E$, $C^s_t(\Alg,A,E) \geq \alpha\cdot C^1_t(X,A,E) + \Delta_C$, where parameter $\Delta_C$ does not depend on $t,X,A$ or $E$, and $\alpha$ is the completed-load competitive ratio, which we denote by $\cC(\Alg)$.
What is more, $\alpha$ is also independent of $t,X,A$ and $E$, but it may depend on system parameters like $\lmin, \lmax, m$ or $s$, which are not considered as inputs of the problem; they are fixed and given upfront. The input of the problem is formed only by the adversarial arrival and error patterns $A$ and $E$.
Finally, let us clarify that the number of machines $m$ is {\em fixed} for a given execution, and that the algorithm used may take it into consideration; hence different $m$ may result to different performance of the same algorithm, affecting additive term in the competitiveness. The same holds for $\lmin,\lmax$ and $s$.

\section{No Speedup}
\label{sec:noSpeedup}

Let us start with the section in which machines have no speedup, i.e., $s=1$. We aim to show that the upper bound of completed-load competitiveness shown in~\cite{Anta2013measuring} can be achieved by some online algorithms in the distributed setting of $m$ machines. In particular, we propose three scheduling algorithms, $\SLBurst$, $\Amortized$ and M$\Amortized$, and analyze their performance under worst-case arrival and error patterns $A$ and $E$, showing that the upper bound of completed-load competitiveness is guaranteed.


\subsection{Two Task Sizes}
\label{subsec:slBurst}

Let us start with the first algorithm, $\SLBurst$, that runs on the $m$ machines of the system and considers only two different task sizes, $\lmin$ and $\lmax$ (see the algorithm's pseudocode Alg.~\ref{algRBurst}).\\

\lstset{columns=fullflexible,
basicstyle=\small,
	tabsize=3,
	columns=flexible,
	mathescape=true,
	morekeywords={Parameters, Calculate, Repeat, Get, Create, Sort, Case, 1,2,3,4, If, then, else, Inform, Upon, awaking, or, restart},
}
\lstset{escapeinside={('}{')}}

\setlength{\textfloatsep}{2pt}
\begin{algorithm}[!t]
\caption{ {\bf $\SLBurst$} (for machine $p$)}
\label{algRBurst}
\begin{lstlisting}
Parameters: $m, \lmin, \lmax$
Upon awaking or restart
	Get $L_\lmin$ and $L_\lmax$, from the Repository;
	$preamble \gets \textbf{FALSE}$	  																												('\footnotesize{//Reset preamble status}')
	$c \gets 0$;																															('\footnotesize{//Reset counter}')
	Calculate $\rhoflr \gets \left\lfloor \frac{\lmax}{\lmin} \right\rfloor$
	If $|L_\lmin| \geq \rhoflr\cdot m^2$ then
		$preamble \gets \textbf{TRUE}$;

	Repeat																												('\footnotesize{//At decision times}')
		Get the queues of pending tasks, $L_\lmin$ and $L_\lmax$, from the Repository;
		Sort $L_\lmin$ and $L_\lmax$ by arrival time (ascending);
		If $preamble = \textbf{TRUE} \wedge (c < \rhoflr)$ then
			execute task at position $p\cdot m$ in $L_\lmin$;
			$c \gets c + 1$;
		else
			If $|L_\lmax| \geq m^2$ then
				execute task at position $p\cdot m$ in $L_\lmax$;
			else if $|L_\lmin| \geq m^2$ then
				execute task at position $p\cdot m$ in $L_\lmin$;
			else if $L_\lmax \neq \emptyset$ then
				execute task at position $(p\cdot m)\mod |L_\lmax|$
				in $L_\lmax$;
			else if $L_\lmin \neq \emptyset$ then
				execute task at position $(p\cdot m)\mod |L_\lmin|$
				in $L_\lmin$;
		Inform the Repository for the task completion;
\end{lstlisting}
\end{algorithm}

{\em Algorithm description.}
Upon awaking or restart, machine $p$ reads the two queues of pending tasks from the Repository, $L_\lmin$ and $L_\lmax$, and applies ascending sort by their arrival time, such that the earliest injected task is at position $1$ of the queue.
It then calculates parameter $\rhoflr = \lfloor \rho \rfloor$ and along with parameter $preamble$ decides which is the next task to be scheduled, avoiding redundancy when {\em enough} tasks are pending. More precisely, at each decision time, if there are at least $\rhoflr\cdot m^2$ tasks of size $\lmin$ and $preamble = \textbf{TRUE}$, the machine attempts to complete $\rhoflr$ $\lmin$-tasks before continuing with a {\em non-redundant} version of the {\em Largest Size} (LS) scheduling approach. Let us explain further: after the preamble is completed by the machine ({\em if} there was enough time and the machine did not crash), it gives priority to the largest tasks, given that it has {\em enough} of them, so that redundancy is avoided (see exact conditions in Algorithm~\ref{algRBurst}). Note that if there are at least $m^2$ $\lmin$-tasks and/or at least $m^2$ of $\lmax$-tasks, each machine will complete a different task of the same size. Hence, if there are not enough $\lmax$-task but there are enough $\lmin$ ones, it will schedule $\lmin$-tasks instead of risking redundant executions with scheduling $\lmax$ ones.\\

Observe that algorithm $\SLBurst$ belongs to the class of scheduling algorithms named \GroupLIS\xspace (see definition in Section~\ref{sec:model}), which was initially defined by Fern\'andez et al. in~\cite{anta2015online}. They showed that the algorithms in this class, considering speedup $s\geq 1$, do not execute the same task twice, {\em as long as there are enough pending tasks} (i.e., $\geq m^2$), and thus we show that the same holds for $\SLBurst$.
Let us start with the next lemma, that corresponds to the adaption of Lemma 8 in~\cite{anta2015online} for the case of no speedup.

\begin{definition}[\cite{anta2015online}]
A {\bf\em full task execution} of a task $\tau$ is the interval $[t,t^\prime]$, during which a machine $p$ schedules $\tau$ at time $t$ and reports its completion to the repository at $t^\prime$, without stopping its execution within the interval $[t,t^\prime)$.
\end{definition}

\begin{lemma}
\label{l:nonred}
For an algorithm $\Alg$ of type \mbox{\GroupLIS} and a time interval $T$ in which a queue $L_\ell$ has at least $m^2$ pending tasks, any two full task executions (in $T$) by different machines, are executions of different tasks; i.e., the full executions of tasks $\tau_1,\tau_2 \in L_\ell$ by machines $p_1$ and $p_2$ respectively, must have $\tau_1 \neq \tau_2$.
\end{lemma}

For the ease of presentation, let us use letter $A$ to refer to our algorithm, $\SLBurst$, in its analysis.
Let us also define the following two intervals, during which there are sufficient tasks pending in order to guarantee non-redundant task executions (from Lemma~\ref{l:nonred}, we make the observation that follows):
\begin{itemize}
\item [$T^+$:] an interval where $|L_\lmax(A,t)| \geq m^2$, $\forall t \in T^+$
\item [$T^-$:] an interval where $|L_\lmin(A,t)| \geq m^2$, $\forall t \in T^-$
\end{itemize}

\begin{observation}
\label{o:no-red}
All full executions of $\lmax$-tasks in the execution of algorithm $\SLBurst$ within any interval $T^+$ appear exactly once. Similarly, all full executions of $\lmin$-tasks in the execution of $\SLBurst$ within an interval $T^-$ appear exactly once.
\end{observation}

Note now, that there are two possible types for the whole execution of Algorithm $\SLBurst$:
\begin{itemize}
\item[(a)] $\forall t, \exists t'> t$ such that\vspace{0.15em}

$|L_\lmin(A,t')| < \rhoflr m^2 \bigwedge |L_\lmax(A,t')| < m^2$.
\item[(b)] $\exists t$ such that $\forall t'> t$\vspace{0.2em}

$|L_\lmin(A,t')| \geq \rhoflr m^2 \bigvee |L_\lmax(A,t')| \geq m^2$.
\end{itemize}

In the first case, when an execution is of type (a), there will always be a time $t'$ after the current time instant $t$, at which the queue $L_\lmin$ has less than $\rhoflr m^2$ tasks and the queue $L_\lmax$ less than $m^2$ tasks. In this case, non-redundancy cannot be guaranteed.

In the second case, when an execution is of type (b), after time $t$ the queue of pending tasks will never become empty, it will instead have {\em enough} pending tasks in order to guarantee non-redundancy at all times. The execution after time instant $t$ can be described by a sequence of intervals, say $T_i$, where $i$ simply denotes their sequence. They are $T^+$ and/or $T^-$ intervals. More precisely, they belong to one of the following types:
\begin{itemize}
\item[(1)] $|L_\lmin(A,t^*)| \geq \rhoflr m^2 \bigwedge |L_\lmax(A,t^*)| \geq m^2$, $\forall t^*\in T$\vspace{0.15em}

In this case, a machine following algorithm $\SLBurst$ will schedule $\rhoflr$ $\lmin$-tasks, followed by continuously scheduled $\lmax$-tasks, until a time instant $t$ where either $|L_\lmin(A,t)| < \rhoflr m^2$ or $|L_\lmax(A,t)| < m^2$, thus one of the next two types of period will follow.
\item[(2)] $|L_\lmin(A,t^*)| \geq \rhoflr m^2 \bigwedge |L_\lmax(A,t^*)| < m^2$, $\forall t^*\in T$\vspace{0.15em}

In this case, a machine following algorithm $\SLBurst$ will continuously schedule $\lmin$-tasks until the queues are such that one of the other two types of periods follow.
\item[(3)] $|L_\lmin(A,t^*)| < \rhoflr m^2 \bigwedge |L_\lmax(A,t^*)| \geq m^2$, $\forall t^*\in T$\vspace{0.15em}

In the third case, machines following the $\SLBurst$ continuously schedule $\lmax$-tasks, until the queues become such that one the previous two types of periods follow.
\end{itemize} 

\begin{lemma}
\label{l:typeA}
For executions of type (a), where $\forall t, \exists t' > t$ s.t. $|L_\lmin(A,t')| < \rhoflr m^2 \bigwedge |L_\lmax(A,t')| < m^2$ holds, the completed-load competitive ratio of algorithm $\SLBurst$ goes to $1$, i.e., $\cC(\SLBurst) = 1$.
\end{lemma}

\begin{proof}
First, let us fix a pair of arrival and error patterns, such that executions of case (a) occur.
We focus on time instant $t'$ from the definition.
Observe that at time $t'$, the total pending load of algorithm $\SLBurst$ is less than $\rhoflr m^2\lmin + m^2\lmax$, while the total pending load of $X$ is at least zero. Let us overload the notation of the set of injected tasks up to time $t'$ such as to represent the total injected {\em load} up to that time; $I_{t'}$ will now represent the sum of sizes of all injected tasks up to time $t'$. Then, at time instant $t'$ the completed load ratio is
\begin{eqnarray*}
\cC_{t'}(\SLBurst) & = & \frac{C_{t'}(\SLBurst)}{C_{t'}(X)} =  \frac{I_{t'} - \big(|L_\lmin(A,t')|\lmin + |L_\lmax(A,t')|\lmax\big)}{I_{t'} - \big(|L_\lmin(X,t')|\lmin + |L_\lmax(X,t')|\lmax\big)} \\
	& \geq & \frac{I_{t'} - (\rhoflr m^2\lmin + m^2 \lmax)}{I_{t'} - \big(|L_\lmin(X,t')|\lmin + |L_\lmax(X,t')|\lmax\big)}
\end{eqnarray*}

which leads to a completed-load competitive ratio of $1$ as time goes to infinity; the total injected load goes to infinity as well:
$$
\cC(\SLBurst) = \lim\limits_{t\rightarrow\infty} \cC_t(\SLBurst) = \lim\limits_{t\rightarrow\infty} \bigg( 1 - \frac{\rhoflr m^2\lmin + m^2\lmax}{I_t} \bigg) = 1.
$$
This completes the proof of the completed load competitiveness for executions of type (a) as claimed.
\end{proof}


The next lemma follows mostly the idea of analysis of algorithm $\Preamble$, presented in~\cite{Anta2013measuring}, for the case of packet scheduling over one communication link. The complete proof is included in the Appendix.

\begin{lemma}
\label{l:typeB}
For executions of type (b), where $\exists t, \forall t' > t$, s.t. $|L_\lmin(A,t')| \geq \rhoflr m^2 \bigvee |L_\lmax(A,t')| \geq m^2$ holds, the completed load competitive ratio of algorithm $\SLBurst$ is $\cC(\SLBurst) \geq \frac{\rhoflr}{\rho + \rhoflr}$.
\end{lemma}


From Lemmas~\ref{l:typeA} and~\ref{l:typeB}, that analyze the two types of executions, (a) and (b) respectively, we have the lower bound for the completed-load competitiveness of $\SLBurst$, given by the following theorem.

\begin{theorem}
\label{t:rmpreamble}
When algorithm $\SLBurst$ runs without speedup ($s=1$) under any arrival and error patterns $A$ and $E$, it has a completed-load competitive ratio $\cC(\SLBurst,A,E) \geq \frac{\rhoflr}{\rho + \rhoflr}$.
\end{theorem}

\subsection{Finite Task Sizes -- Pairwise Divisible}
\label{subsec:amortized}

We now move to the case of $k>2$ different task sizes. Let us denote them by $\lmin = \ell_1 < \ell_2 < \dots < \ell_k = \lmax$. We assume that each ratio $\rho_{i,j} = \ell_i/\ell_j$ is an integer for any $1\leq j < i \leq k$, a property of the task sizes called {\em pairwise divisibility}.

Inspired by the work of Jurdzinski et al.~\cite{Jurdzinski2015}, we propose algorithm \Amortized, that uses this property to schedule the tasks among the $m$ machines of the system, and analyze its completed-load competitiveness when run without speedup (see the algorithm's pseudocode in Alg.~\ref{algAmort} and~\ref{alg:schgroup}).
The algorithm follows the {\em Shortest Size} (SS) first policy, but subject to some balancing constraints. It is based on scheduling tasks in batches (or groups) that balance the length of the next larger task. What is more, it considers redundancy avoidance, demanding {\em enough} tasks available before scheduling.
An important difference with algorithm $\SLBurst$, apart from the fact that this one considers $k$ different task sizes, is that it continuously schedules bursts of short tasks before going to the next larger task (if the tasks are available and the machines do not crash).\\

\lstset{numbers=left,
basicstyle=\small,
	numbersep=4pt,
	columns=fullflexible,
	mathescape=true,
	morekeywords={Parameters, Repeat, While, Do, If, then, Return, Upon, awaking, or, restart},
}
\lstset{escapeinside={('}{')}}

\setlength{\textfloatsep}{2pt}
\begin{algorithm}[!t]
\caption{ {\bf $\Amortized$} (for machine $p$)}
\label{algAmort}
\begin{lstlisting}
Parameters: $m, \{\ell_1, \ell_2,\dots, \ell_k\}$
Upon awaking or restart
	Repeat
		Get $L_1$ to $L_k$ from the Repository;
		While $\ell_k \Big\lfloor \frac{|L_k|}{m^2} \Big\rfloor + \sum\limits_{i=1}^{k-1} \ell_i \Big\lfloor \frac{|L_i|}{m^2 + m\rho_{i+1}} \Big\rfloor < \ell_k $ Do
			execute task $\ell$ at position $(p\cdot m) \mod |Q|$ in $Q$;
			Inform the Repository for the task completion;
		$Schedule\_Group(k)$;
\end{lstlisting}
\end{algorithm}\vspace{-1.5em}
\begin{algorithm}[!t]
\caption{ {\bf $Schedule\_Group(j)$}}
\label{alg:schgroup}
\begin{lstlisting}
Parameters: $m, j, \{\ell_1,\ell_2,\dots,\ell_k\}, \{L_1,L_2,\dots,L_j\}$
	If $\sum\limits_{i=1}^{j-1} \ell_i \Big\lfloor \frac{|L_i|}{m^2 + m\rho_{i+1}} \Big\rfloor \geq \ell_j $ then
		For $\alpha = 1$ to $\rho_j$ Do
			Schedule_Group(j-1);
	Else
		execute task $\ell_j$ at position $p\cdot m$ in $L_j$;
		Inform the Repository for the task completion;
	Return
\end{lstlisting}
\end{algorithm}

%

\paragraph{Special Notation \& Terminology}
Before looking at the details of the algorithm, let us introduce some necessary notation and terminology that will be used extensively in this subsection.
%
First, parameter $\rho_{i,j} = \ell_i/\ell_j$, where $1\leq j < i \leq k$, as already mentioned at the beginning of the subsection, represents the ratio between two task sizes and is considered to be an integer for this algorithm. A special case of this ratio used in the algorithm, is $\rho_i = \frac{\ell_i}{\ell_{i-1}}$, where $i \in [2,k]$; it represents the ratio between two consecutive task sizes.

\begin{definition}
\label{d:adequate}
We define {\em adequate sizes} of pending tasks, the task sizes whose pending queues have $\gtrsim m^2$ tasks. More precisely, for size $\ell_k$ to be adequate there must be at least $m^2$ tasks pending in the $L_k$ queue, while for any other $\ell_j$ size, where $j\in [1,k-1]$, the corresponding pending queue $L_j$ must have at least $m^2 + m\rho_{j+1}$ tasks. (In Lemma~\ref{l:enough}, we show that this is the necessary number of tasks in order to guarantee the non-redundancy property of the algorithm.)
\end{definition}

\begin{definition}
\label{d:igroup}
We define an {\em $i$-group of tasks}, where $i \in [1,k]$, being the tasks completed in the execution of a machine under the recursive call to the function {\em Schedule\_Group$(i)$}. Note, that an $i$-group has a total size of $\ell_i$, but may be the result of the completion of several smaller tasks.
\end{definition}

\begin{definition}
We consider a machine to be {\em busy} at time $t$ if it is either executing some task, it has just completed one, or it is crashed (it is either the time instant that the machine was just crashed, or the machine has been crashed for some time). Observe that an interval from a crash to a restart instant, say $T = [t_c, t_r]$, belongs to the busy interval; during that interval, no algorithm is able to complete any task, hence it does not affect the completed load. Otherwise, it is considered to be {\em idle}.\\
We also consider a machine to be in an {\em $n$-busy} interval, say $T = [t_1,t_2]$, where $t_1<t_2$, satisfying the following properties:\\
\noindent (1) The machine is busy at each time $t\in T$.\\
\noindent (2) The machine does not schedule tasks of size $\ell_i$ for $i>n$ during the interval $T$.\\
\noindent (3) At the beginning of interval $T$, i.e., time instant $t_1$, algorithm \Amortized has at least as many tasks of size $\ell_i$ pending as $X$, for each $i \leq n$. Hence, $P^A(t_1,\ell_i) \geq P^X(t_1,\ell_i)$ for each $i\in[1,n]$.
\end{definition}

Finally, note that in the pseudo-code we refer to queue of pending tasks $Q$. Recall that in Section~\ref{sec:model} we define $Q$ to be the set of pending tasks. Here, we overload its definition to make this set a unified queue of all $L_i$, sorted in an ascending order according to the task sizes.\\

{\em Algorithm description.}
After awakening or restart, a machine schedules the task at position $p\cdot m$ of the pending queue $Q$, until the sum of the {\em adequate sizes} of the pending tasks is at least $\ell_k$.
Following this strategy, the algorithm guarantees the ability to cover a time interval of length $\ell_k$, with non-redundant task executions, if it is not interrupted by a machine crash. Otherwise, being work-conserving, it schedules the task in the position already mentioned, but with no guarantees of non redundancy.
Then, it calls the recursive function {\em Schedule\_Group$(j)$} (starting with $j=k$) which checks whether the sum of {\em adequate sizes} of the pending tasks smaller than $\ell_j$ is at least equal to $\ell_j$ (resp., $\ell_k$).
If the condition is true, the algorithm makes  $\rho_j = \frac{\ell_j}{\ell_{j-1}}$ calls to function {\em Schedule\_Group$(j-1)$} (resp., {\em Schedule\_Group$(k-1)$}) in order to cover the corresponding time interval $\ell_j$ with $\rho_j$ groups of $\ell_{j-1}$ aggregate size; in other words, $\rho_j$ $(j-1)$-groups. In the following recursion levels more recursive calls may occur, if there are again enough pending tasks, thus covering the corresponding time intervals by tasks of smaller size each time. 
Otherwise, when there are not enough tasks pending in a recursion level, a task of the current size, $\ell_j$, is scheduled by the machine and when completed returns to the previous recursion level.\\


We will now analyze algorithm \Amortized, proving some important properties and showing that its completed-load competitiveness is indeed optimal, i.e., $\cC(\Amortized) \gtrsim 1/2$. We start with two lemmas that lead to the {\em non-redundancy} property of the algorithm, omitting the case when lines 6 and 7 of Algorithm~\ref{algAmort} are executed. Note that in that case, the pending load of the algorithm is bounded, so it does not affect the completed-load competitiveness in the long run.

\begin{lemma}
\label{l:enough}
Algorithm \Amortized schedules a $\ell_j$-task (in line 6 of Alg.~\ref{alg:schgroup}), only when there are at least $m^2$ tasks in the corresponding queue of pending tasks, $L_j$.
\end{lemma}

\begin{proof}
Let us start by looking at the algorithm description and its pseudo-code. The first call to schedule some tasks -- calling function {\em Schedule\_Group$(k)$} -- is done only if {\em enough} tasks are pending to cover the $\ell_k$-time without redundancy. A task size is accounted for only when it is {\em adequate}; only when there are $\gtrsim m^2$ tasks of that size pending (recall Definition~\ref{d:adequate}). 

Then, within the {\em Schedule\_Group$(j)$} function, starting by $j = k$, the algorithm checks whether it can be covered non-redundantly by tasks of smaller size. If it does, it makes a recursive call to the function with parameter $j-1$, which corresponds to the next smaller task size, $\ell_{j-1}$.
A task $\ell_j$ cannot be covered non-redundantly by smaller tasks when $\sum\limits_{i=1}^{j-1} \ell_i \Big\lfloor \frac{|L_i|}{m^2 + m\rho_{i+1}} \Big\rfloor < \ell_j$ (see condition of line 2 in Algorithm~\ref{alg:schgroup}).
However, a function call with task size $\ell_j$ means that it was either called by algorithm \Amortized directly (and $j=k$), in which case the condition in line 8 of Alg.~\ref{algAmort} does not hold, or it was called by the previous recursion level; by function handling the next bigger task size, $\ell_{j+1}$, in which case the condition in line 2 holds. We will now show that in either case, there are {\em enough} tasks of size $\ell_j$ to be scheduled by the system's machines without executing any of them redundantly.


We consider the function call {\em Schedule\_Group$(j)$}, for which the {\em if} condition in line 2 does not hold. This implies
\begin{equation}
\label{e:eq1}
\sum\limits_{i=1}^{j-1} \ell_i \Big\lfloor \frac{|L_i|}{m^2 + m\rho_{i+1}} \Big\rfloor < \ell_j.
\end{equation}

We then consider and analyze the two cases mentioned above, separately:\vspace{0.5em}

\indent {\bf Case 1:} A previous function call, {\em Schedule\_Group$(j+1)$}, in which the {\em if} condition in line 2 holds, i.e., $\sum\limits_{i=1}^j \ell_i \Big\lfloor \frac{|L_i|}{m^2 + m\rho_{i+1}} \Big\rfloor \geq \ell_{j+1}$, implies that
\begin{equation}
\label{e:eq2}
\ell_j \Big\lfloor \frac{|L_j|}{m^2 + m\rho_{j+1}} \Big\rfloor + \sum\limits_{i=1}^{j-1} \ell_i \Big\lfloor \frac{|L_i|}{m^2 + m\rho_{i+1}} \Big\rfloor \geq \ell_{j+1}
\end{equation}
Combining the two equations~\ref{e:eq1} and~\ref{e:eq2} we have the following
\begin{eqnarray*}
\ell_j \Big\lfloor \frac{|L_j|}{m^2 + m\rho_{j+1}} \Big\rfloor & \geq & \ell_{j+1} - \sum\limits_{i=1}^{j-1} \ell_i \Big\lfloor \frac{|L_i|}{m^2 + m\rho_{i+1}} \Big\rfloor > \ell_{j+1} - \ell_j > 0 \\
\Rightarrow \Big\lfloor \frac{|L_j|}{m^2 + m\rho_{j+1}} \Big\rfloor  & > & 0,
\end{eqnarray*}
which means that $|L_j| \geq m^2 + m\rho_{j+1}$.\vspace{0.5em}

{\bf Case 2:} The function call {\em Schedule\_Group$(j)$} was actually {\em Schedule\_Group$(k)$} and came directly from line 8 of algorithm \Amortized. Hence, $\ell_k \Big\lfloor \frac{|L_k|}{m^2} \Big\rfloor + \sum\limits_{i=1}^{k-1} \ell_i \Big\lfloor \frac{|L_i|}{m^2 + m\rho_{i+1}} \Big\rfloor \geq \ell_k$ holds, which implies that
\begin{equation}
\label{e:eq3}
\ell_k \Big\lfloor \frac{|L_k|}{m^2} \Big\rfloor \geq \ell_k - \sum\limits_{i=1}^{k-1} \ell_i \Big\lfloor \frac{|L_i|}{m^2 + m\rho_{i+1}} \Big\rfloor
\end{equation}

Replacing $j = k$ in equation~\ref{e:eq1}, we have $\sum\limits_{i=1}^{k-1} \ell_i \Big\lfloor \frac{|L_i|}{m^2 + m\rho_{i+1}} \Big\rfloor < \ell_k$, which combined with equation~\ref{e:eq3} we can easily see that
\begin{equation*}
\ell_k \Big\lfloor \frac{|L_k|}{m^2} \Big\rfloor > 0 \Rightarrow \Big\lfloor \frac{|L_k|}{m^2} \Big\rfloor > 0,
\end{equation*}
which in its turn means that $|L_k| \geq m^2$.\vspace{0.5em}

In both cases, we have seen that there are at least $m^2$ tasks of size $\ell_j$ or $\ell_k$ respectively. However, in case 1 above, there will be $\rho_{j+1}$ iterations of the recursive function call {\em Schedule\_Group$(j)$} of line 4. We must therefore make sure that there are at least $m^2$ available tasks for all iterations. 

Consider for example, the case in which at a time $t$ all $m$ machines are in a $(j+1)$-group execution; following the {\em Schedule\_Group$(j+1)$} function, and having condition of line 2 \textsc{true}. Then, they all start the $\rho_{j+1}$ iterations of scheduling $j$-groups of tasks, calling the recursive function {\em Schedule\_Group$(j)$}. Consider now, that in all corresponding conditions of line 2, are \textsc{false}. This means, that all $m$ machines will simultaneously execute one $\ell_j$-task in every iteration. Therefore, having $m^2 + m\rho_{j+1}$ pending tasks of size $\ell_j$ at the beginning of iterations, will guarantee that in every iteration there are still at least $m^2$ tasks pending in queue $L_j$. This completes the proof of the lemma.
\end{proof}

Observe now, that algorithm $\SLBurst$ is of type \GroupLIS; it separates the pending tasks into classes depending on their size, it sorts them with respect to their arrival time, and if a class contains at least $m^2$ pending tasks, a machine $p$ schedules the task at position $(p\cdot m)$. Hence, Lemma~\ref{l:nonred} also holds for it.
Hence, combining the two Lemmas~\ref{l:enough} and~\ref{l:nonred}, the following property for algorithm \Amortized follows.
\begin{observation}
\label{o:property1}
When Algorithm \Amortized, schedules tasks through its function Schedule\_Group (Alg.~\ref{alg:schgroup}), it never completes the same task more than once. In particular, the same task cannot be simultaneously executed in more than one machines of the system.
\end{observation}

\begin{lemma}
\label{l:pending}
When a task of size $\ell_j$ is scheduled by \Amortized, through its function {Schedule\_Group}, say at time instant $t$, the total size of smaller pending tasks is $P^A(t,<\ell_j) \leq \sum\limits_{i=1}^{j-1} (\ell_j + \ell_i)(m^2 + m\rho_{i+1})$.
\end{lemma}

\begin{proof}
When a task of size $\ell_j$ is scheduled by \Amortized in line 6 of Alg.~\ref{alg:schgroup}, as we have seen also from Lemma~\ref{l:enough}, the following inequality must hold: $\sum\limits_{i=1}^{j-1}\ell_i \Big\lfloor \frac{|L_i|}{m^2 + m\rho_{i+1}} \Big\rfloor < \ell_j$. This also means that $\forall i \in [1,j-1]$, the following is true:
$$
\ell_i  \left\lfloor \frac{|L_i|}{m^2 + m\rho_{i+1}} \right\rfloor < \ell_j \Rightarrow \ell_i  \left( \frac{|L_i|}{m^2 + m\rho_{i+1}} - 1 \right) < \ell_j
 \Rightarrow \ell_i |L_i| < (\ell_j + \ell_i)(m^2 + m\rho_{i+1}).
$$
Therefore, the sum of all pending tasks smaller than $\ell_j$ is
$
P^A(t,<\ell_j) = \sum\limits_{i=1}^{j-1} \ell_i |L_i| \leq \sum\limits_{i=1}^{j-1} (\ell_j + \ell_i)(m^2 + m\rho_{i+1})
$
as claimed.
\end{proof}


By Observation~\ref{o:property1}, we have that no task is executed more than once by algorithm \Amortized, when scheduled by its function {\em Schedule\_Group}. Hence, we can safely separate the analysis of each machine individually, safely ignoring any task execution by line $6$ of Alg.~\ref{algAmort}. We focus on one machine, say $p$, and then generalizing for all $m$ machines to give the final result. We look at the machine's $n$-busy intervals, for some $n \leq k$, and prove a completed load competitiveness of $1/2$, provided that $\rho_{i,j} \in \mathbb{N}$ for $1 \leq j < i \leq k$, which is in fact optimal. {\em (The omitted proofs can be found in the Appendix).}

\begin{lemma}
\label{l:l1}
For a machine $p$ that is $n$-busy at a time interval $T$, where $n \leq k$, its total completed load with \Amortized is at least as large as its completed load with $X$ accounting only tasks of size $\geq \ell_n$, in interval $T$, minus $\ell_k$; i.e., $C_p^A(T) \geq C_p^X(T,\geq \ell_n) - \ell_k$.
\end{lemma}

\begin{lemma}
\label{l:l2}
For a machine $p$ that is $n$-busy at a time interval $T = [t_1,t_2]$, where $n \leq k$, assume time $t\in T$ be any time when \Amortized starts executing a $\ell_n$-task. Then,
$$
2C_p^A\big([t_1,t]\big) \geq C_p^X\big([t_1,t]\big) + P^X(t,<\ell_n) - \sum\limits_{i=1}^{n-1}(\ell_n + \ell_i)(m^2 + m\rho_{i+1}) - \ell_k.
$$
\end{lemma}

\begin{lemma}
\label{l:l3}
Let $f_n$ be a function such that $f_1 = \ell_k$ and $f_{i+1} = f_i + \sum\limits_{j=1}^{i}(\ell_{i+1} + \ell_j)(m^2 + m\rho_{j+1}) + \ell_{i+1} + 2\ell_k$.
For a machine $p$ that is $n$-busy at a time interval $T$, where $n \leq k$, $$2C_p^A(T) \geq C_p^X(T) - f_n.$$
\end{lemma}

\begin{theorem}
\label{th:Amortized}
Algorithm \Amortized has an optimal completed load competitiveness of 1/2, provided that $\ell_i/\ell_{i-1}\in \mathbb{N}$ for each $i\in [2,k]$.
\end{theorem}

\begin{proof}
First, by Observation~\ref{o:property1}, we know that each task completion within function {\em Schedule\_Group}, occurs only once. Then, looking at Lemma~\ref{l:l3}, it implies that the completed-load competitiveness of algorithm \Amortized gets arbitrarily close to 1/2 on sufficiently long periods of time in which it is busy and $X$ starts with the queue containing at most the same tasks of each size.
On the other hand, \Amortized cannot guarantee non redundancy when its queue contains few tasks, i.e., when $\ell_k \Big\lfloor \frac{|L_k|}{m^2} \Big\rfloor + \sum\limits_{i=1}^{k-1} \ell_i \Big\lfloor \frac{|L_i|}{m^2 + m\rho_{i+1}} \Big\rfloor < \ell_k $.
This means, that as time goes to infinity, the completed load competitiveness goes to 1/2 as claimed.

In~\cite{Anta2013measuring}, Fern\'andez Anta et al. showed that the completed load of any online algorithm for two different task sizes is at most $\frac{\bar{\rho}}{\bar{\rho} + \rho}$, which is equal to $1/2$ when $\rho \in \mathbb{N}$. Hence, since an adversary can decide to schedule merely two different task sizes among the available $k$ ones, it means that the completed load competitiveness shown is in fact optimal.
\end{proof}


\subsection{Finite Task Sizes -- General}
\label{subsec:Mamortized}

We now look at the case when $\rho_{i,j} = \ell_i/\ell_j \not\in \mathbb{N}$.
Theorem 2 in~\cite{Jurdzinski2015} shows that the completed load competitive ratio of any scheduling algorithm running without speedup, is at most
$\min\limits_{i\leq j<i\leq k} \Big\{\frac{\overline{\rho_{i,j}}}{\overline{\rho_{i,j}} + \rho_{i,j}}\Big\}$. This upper bound also holds in the case of multiple machines, since the adversary can force only one of the machines to be alive at all times, while keeping the rest crashed.

However, if algorithm \Amortized is used in this case, the additional advantage of an offline algorithm $X$ from rounding on each recursion level, can accumulate and worsen the competitiveness ratio proven in the previous section. Thus, we must present a modified algorithm, that can reach the upper bound mentioned and hence be optimal.

Let us assume a modification such that instead of executing $\rho_i$ $(i-1)$-groups on the recursion level $i-1$, the algorithm keeps sending $(i-1)$-groups while the completed load of the tasks in the groups completed are less than $\ell_i - \ell_{i-1}$. We can then show the following.

\begin{claim}
\label{c:mod1}
The modified algorithm described above has a completed load competitiveness ratio at least $\min_{i\in[2,k]}\big\{\frac{\rho_i - 1}{2\rho_i - 1}\big\}$.
\end{claim}
\begin{proof}
Looking at Lemmas~\ref{l:l1},~\ref{l:l2} and~\ref{l:l3}, and $n$-busy intervals, we use the association of each $\ell_i$-task, for $i\geq n$, completed by $X$, with $\rho_{i,n}$ $n$-groups completed by \Amortized at the same interval, hence covering the $\ell_i$ completed load. For the modified algorithm, a task $\ell_i$ completed by $X$ would correspond to a group of tasks of aggregate size at least $\ell_i - \ell_{i-1}$, which would translate to the following three equations:
\begin{equation*}
\frac{\rho_n}{\rho_n - 1} C^A(T) \geq C^X(T,\geq \ell_n) - \ell_k
\end{equation*}
\begin{equation*}
\left(i+\frac{\rho_n}{\rho_n - 1}\right) C^A\big([t_1,t]\big) \geq C^X\big([t_1,t]\big) + P^X(t,< \ell_n)	- \sum\limits_{i=1}^n (\ell_{n+1} + \ell_i)(m^2 + m\rho_{i+1}) - \ell_n
\end{equation*}
\begin{equation*}
\left( 1 + \frac{\rho_n}{\rho_n - 1} \right) C^A(T) \geq C^X(T) - f_n
\end{equation*}
Applying these inequalities in the proof of Theorem~\ref{th:Amortized} we have the result claimed.

Note here, that since $\rho_i$ is not an integer, for every $\rho_i \geq \rho_n$ task completed by $X$, it could be the case that more than $\lfloor \rho_i/\rho_n\rfloor$ groups of length $\rho_n$ are necessary to cover the execution time of $\rho_i$.
\end{proof}

Nonetheless, the tasks completed in each $i$-group, for $i\in[2,k]$ might be of different sizes, making the comparison of the completed load competitiveness of the above claim with the completed load competitiveness of Theorem~\ref{th:Amortized} ambiguous. Another modification is therefore necessary in order to tackle this uncertainty, and we present it with algorithm M$\Amortized$ (see pseudo-code in Alg.~\ref{algMAmort}).\\

{\em Algorithm description.}
Algorithm M\Amortized completes tasks of the same size as long as possible and change only when it is necessary. For that, we split the execution into {\em stages} of total length $ck\ell_k$, where $c\in\mathbb{N}$ is a fixed large constant. At the beginning of a stage, a set of {\em candidate} task sizes is established as $\cc = \Big\{i \Big| \ell_i \big\lfloor \frac{|L_i|}{m^2} \big\rfloor \geq ck\ell_k \Big\}$. Then, the {\em appropriate size} $\ell_{i^*}$ is set as the minimum size in the set of candidate sizes, i.e., $i^* = \min(\cc)$. This is the size of tasks that the algorithm will start executing. The appropriate size is updated after every task completion, checking first whether there has been some change in the set of candidate tasks -- due to new task injections.
More details are given by the pseudo-code in Algorithms~\ref{algMAmort} and~\ref{Mschgroup}.\\

\lstset{numbers=left,
basicstyle=\small,
	numbersep=4pt,
	columns=fullflexible,
	mathescape=true,
	morekeywords={Parameters, Loop, While, Do, If, Else, then, Return, Upon, awaking, or, restart},
}
\lstset{escapeinside={('}{')}}

\setlength{\textfloatsep}{2pt}
\begin{algorithm}[!t]
\caption{ {\bf M$\Amortized$} (for machine $p$)}
\label{algMAmort}
\begin{lstlisting}
Parameters: $m, \{\ell_1, \ell_2,\dots, \ell_k\}$
Upon awaking or restart
	Repeat
		Get all pending queues from the Repository, $L_i$;
		$C \gets \Big\{i \Big| \ell_i\Big\lfloor \frac{|L_i|}{m^2} \Big\rfloor \geq ck\ell_k \Big\}$;
		While $\Big\{i \Big| \ell_i\Big\lfloor \frac{|L_i|}{m^2} \Big\rfloor \geq ck\ell_k \Big\} = \emptyset$ Do
			execute task $\ell$ at position $(p\cdot m) \mod |Q|$ in $Q$;
			Inform the Repository for the task completion;
		$C \gets \Big\{i \Big| \ell_i\Big\lfloor \frac{|L_i|}{m^2} \Big\rfloor \geq ck\ell_k \Big\}$;
		$i^* \gets \min(C)$;
		For $a = 1$ to $ck$ Do
			$\ell' \gets Schedule\_Group(k)$;
\end{lstlisting}
\end{algorithm}
\begin{algorithm}[!t]
\caption{ {\bf $Schedule\_Group(j)$}}
\label{Mschgroup}
\begin{lstlisting}
Parameters: $m, j, \{\ell_1,\ell_2,\dots,\ell_k\}, \{L_1,L_2,\dots,L_j\}$
	$\ell \gets 0$;
	While $\ell \leq \ell_j - \ell_{i^*}$ Do
		If $j > i^*$ then
			$\ell' \gets Schedule\_Group(j-1)$;
			$\ell \gets \ell + \ell'$;
		Else
			execute task $\ell_j$ at position $p\cdot m$ in $L_j$;
			If task $\ell_j$ completed successfully then
				Inform the Repository for the task completion;
				$\ell \gets \ell_j$;
			$C \gets C \cup \Big\{i \Big| \ell_i\Big\lfloor \frac{|L_i|}{m^2} \Big\rfloor \geq ck\ell_k \Big\}$;
			$i^* \gets \min(C)$;
	Return $\ell$
\end{lstlisting}
\end{algorithm}

%

As a first observation, let us clearly note that following algorithm M\Amortized, only tasks of size $\ell_{i^*}$ are scheduled, unless there are not enough tasks to guarantee non-redundancy (when set $\cc = \emptyset$), in which case, a task $\ell$ at position $p\cdot m$ of the whole queue $Q$ is scheduled. To see this clearly, look at lines 6-12 in Alg.~\ref{algMAmort} and lines 4,7 and 8 in Alg.~\ref{Mschgroup}.
It is also important to note that parameter $i^*$ can only decrease in each stage. This is because at the beginning of the stage there are enough pending $\ell_i$-tasks for each candidate task size $i \in \cc$ to cover a time interval of length $ck\ell_k$ (line 9 of Alg.~\ref{algMAmort}).
Also, like algorithm \Amortized, the modified algorithm M\Amortized belongs to the $\GroupLIS$\xspace algorithms and has the property of {\em non redundancy} when enough tasks are pending, which we show in the following lemma.

\begin{lemma}
Algorithm M\Amortized never completes the same task more than once within Schedule\_Group (Alg.~\ref{Mschgroup}).
\end{lemma}
\begin{proof}
Let us begin by showing that the algorithm schedules a $\ell_j$-task in function {\em Schedule\_Group} only when there are at least $m^2$ tasks in the corresponding pending queue $L_j$.
Looking at the pseudo-code, a task $\ell_j$ is scheduled in line 8 of Alg.~\ref{Mschgroup}, only when the condition in line 4 does not hold, and hence $j=i^*$. From lines 9 and 10 of Alg.~\ref{algMAmort} and lines 12 and 13 of Alg.~\ref{Mschgroup}, we know that $i^*$ belongs to the set of {\em candidate} task sizes, for which every task size has at least as many tasks pending as necessary to ``cover" $ck\ell_k$ time, i.e., $ck$ calls to the {\em Schedule\_Group($j$)} function. This number of tasks is:
$$
\ell_j\left\lfloor \frac{|L_j|}{m^2} \right\rfloor \geq ck\ell_k \Rightarrow
\left\lfloor \frac{|L_j|}{m^2} \right\rfloor \geq ck\big\lfloor\rho_{k,i}\big\rfloor 	\Rightarrow |L_i| \geq \big(ck\big\lfloor\rho_{k,i}\big\rfloor + 1\big)m^2.
$$
What is more, since $|L_i| \geq m^2$, Lemma~\ref{l:nonred} holds for algorithm M\Amortized as well (it belongs to the $GroupLIS$ algorithms), and hence combining the two, we have the property claimed.
\end{proof}

%
%

\remove{
\begin{definition}
An execution of {\em Schedule\_Group($k$)} is considered to be {\em uniform}, if the algorithm completes tasks of a single fixed size $\ell_i$ only, during the executions of the current {\em Schedule\_Group($k$)} as well as the previous one.
\end{definition}

\begin{claim}
\label{c:uniform}
At least $ck-2k$ calls of {\em Schedule\_Group($k$)} in a stage of M\Amortized are uniform.
\end{claim}
\begin{proof}
As already mentioned, the value of $i^*$ can only decrease during a stage. Since there are up to $k$ task sizes, this can happen up to $k-1$ times. However, for a {\em Schedule\_Group($k$)} execution to be uniform, its previous execution must also be uniform. Hence, there can be at least $ck - 2k$ calls of {\em Schedule\_Group($k$)} in a stage that are uniform.
\end{proof}
}


\begin{lemma}
\label{l:Mamortized}
Algorithm M\Amortized has completed-load competitiveness at least $\min\limits_{1\leq j < i \leq k} \Big\{ \frac{ \overline{ \rho_{i,j} } }{\overline{ \rho_{i,j} } + \rho_{i,j}} \Big\} \cdot c'$, where $c'$ is a constant that depends on parameter $c$ of the algorithm. For large enough $c$, $c'$ can be arbitrarily close to 1.
\end{lemma}

Consider now adapting the value of $c$ in the executions of algorithm M\Amortized; in particular, gradually increasing it to $2c$ when the total completed load of the executed tasks is big enough to guarantee the current competitive ratio close enough to $c'\gamma$ for the current value of $c$. Following this adaptation, the completed load competitiveness will get arbitrarily close to $\gamma$ after sufficiently long time, giving the following theorem.

\begin{theorem}
\label{th:mkAmortized}
Algorithm M\Amortized can reach the optimal completed-load competitiveness, $\min\limits_{1\leq j < i \leq k} \Big\{ \frac{\overline{\rho_{i,j}}}{\overline{\rho_{i,j}} + \rho_{i,j}} \Big\}$.
\end{theorem}


\section{Speedup}
\label{sec:Speedup}

Let us now look at the case in which the machines have speedup $s\geq 1$. As we have mentioned, the negative results (upper bounds) of completed-load competitiveness of any work-conserving algorithm $\Alg_W$ shown for the setting of $1$ machine, still hold for the case of $m$ machines. However, the positive results (lower bounds) may not. In this case, we show that for specific amounts of speedup two positive results are preserved in the multiple machine setting.

\begin{theorem}
\label{t:srho}
Any distributed work-conserving algorithm $\Alg_W$, running on a system with $m$ machines with speedup $s \geq \rho$, that guarantees non redundant executions while there are at least $m^2$ tasks pending, has a completed-load competitive ratio $\cC(\Alg_W) \geq 1/\rho$.
\end{theorem}

\begin{proof}
We consider any distributed work-conserving algorithm $\Alg_W$, running on $m$ parallel machines with speedup $s \geq \rho$.
For the proof of the theorem we consider only the periods of execution during which there are at least $m^2$ pending tasks and look at the number of pending and completed tasks. During the remaining time of the executions, the completed load is bounded by the number of tasks pending (i.e., $< m^2\lmax$).

Let us then consider the execution of each machine of the system individually, and look at the number of completed tasks. In particular, observe first that for time instant $t = 0$, at the beginning of an execution of each machine $p$, $|N_p(A,0)| \geq |N_p(X,0)|$. Then consider any time instant $t>0$ and a corresponding $t'<t$ in its execution, such that $t'$ is the latest time before $t$ that the machine has either crashed or restarted. 
By the definition of $t'$ there are always at least $m^2$ tasks within interval $T = (t',t]$. By induction hypothesis, at time $t'$, $|N_p(A,t')| \geq |N_p(X,t')|$.

Now let $I_T$ be the number of tasks injected during interval $T$. Since $\Alg$ is a work-conserving algorithm, it is continuously scheduling and executing tasks in the interval $T$. What is more, we know that it need at most $\frac{\lmax}{s} \leq \lmin$ time to execute any task, since $s \geq \rho = \frac{\lmax}{\lmin}$. This means that
$|N_p(A,t)| \geq |N_p(A,t')| + \Big\lfloor \frac{t - t'}{\lmax/s} \Big\rfloor \geq |N_p(A,t')| + \Big\lfloor \frac{t - t'}{\lmin} \Big\rfloor$.
On the other hand, algorithm $X$ needs at least $\lmin$ time to complete a task, which means that $|N_p(X,t)| \leq |N_p(X,t')| + \Big\lfloor \frac{t - t'}{\lmin} \Big\rfloor$.
This results to $|N_p(A,t)| \geq |N_p(X,t)|$.

Combining the result for all machines of the system we have that $|N(A,t)| = \sum\limits_{i=1}^m |N_i(A,t)| \geq \sum\limits_{i=1}^m |N_i(X,t)| = |N(X,t)|$. This leads to the desired completed-load competitiveness, $\cC(\Alg_W) \geq 1/\rho$ since $\Alg_W$ may be completing only $\lmin$-tasks while $X$ completes $\lmax$ ones.
\end{proof}

\begin{theorem}
\label{t:srhoplus}
Any distributed work-conserving algorithm $\Alg_W$ running in a system with $m$ machines and speedup $s \geq 1 + \rho$, that guarantees non redundant executions while there are at least $m^2$ tasks pending, has completed-load competitive ratio $\cC(\Alg) \geq 1$.
\end{theorem}

In~\cite{Anta2015competitive}, we studied some of the most popular algorithms in task scheduling, in the setting of one machine and analyzed their complete-load competitiveness under various ranges of speedup. Algorithm $\LIS$ becomes $1$-completed-load competitive as soon as $s \geq \max\{\rho,2\}$. However, when looking at its performance in the setting of $m$ machines (see its pseudo-code Alg.~\ref{algLIS} in the Appendix), we realized that even in the case of $2$ machines, it may not achieve $1$-completed-load competitiveness.

\begin{theorem}
\label{t:mLIS}
When algorithm $\LISs$ runs in a parallel system of two machines ($m=2$) and speedup $s = \rho = 2$, it is not $1$-completed-load competitive, i.e., $\cC(\LISs) < 1$.
\end{theorem}





\section{Conclusions}
\label{sec:conclusions}

In this work, we present the problem of online distributed scheduling of tasks with different computational demands on fault-prone parallel systems. We conduct worst-case analysis of deterministic work-conserving algorithms, looking at their completed-load competitiveness as the performance metric.

We show that the upper bound shown for the case of a single machine and no speedup in~\cite{Anta2013measuring} can be achieved in our setting with $m$ machines, making the result a tight bound. Additionally, the algorithms for scheduling packets of $k$ packet lengths in one link in~\cite{Jurdzinski2015} can also be adapted to task scheduling in one machine, and then non-trivially generalized to $m$ machines. Hence,
we present algorithms for the cases of two or $k$ different task sizes that achieve optimal completed load when run without speedup.

We also show that in the case of speedup, $s > 1$, the competitiveness can be improved. In particular, when speedup is $s \geq 1 +\rho$, any deterministic work-conserving algorithm $\Alg_W$ may achieve optimal completed load $\cC(\Alg) = 1$.
However, we also give a negative result for algorithm $\LISs$, the natural parallel version of the popular {\em Longest In System} scheduling policy. We show, that while
with $1$ machine it achieves $1$-completed-load competitiveness with speedup $s \geq \max\{\rho,2\}$, in a system of two machines running with speedup $s = \rho = 2$ its completed-load competitiveness is $\cC(\LISs) < 1$.

There are still a few open questions though, some of which we would like to answer in future works. We believe that \LISs is an important and interesting algorithm,
for its popularity and fairness. As we just mentioned, the completed-load competitiveness depends on the number of machines $m$. It would be interesting
to understand better the exact relation between $m$, $s$, and completed-load competitiveness in \LISs.
A second concrete question is whether there is a scheduling algorithm
that achieves $1$-completed load competitiveness with speedup $s < 1 + \rho$.


\section{Acknowledgements}
Supported in part by the grant TEC2014-55713-R of the Spanish Ministry of Economy and Competitiveness (MINECO), the Regional Government of Madrid (CM) grant Cloud4BigData (S2013/ICE-2894, co- funded by FSE \& FEDER), the NSF of China grant 61520106005, and European Commission H2020 grants ReCred and NOTRE.
Also partially supported by the FPU12/00505 grant from the Spanish Ministry of Education, Culture and Sports (MECD).

\bibliographystyle{plain}
\bibliography{references}

\appendix

\section{Omitted Proofs}

\subsection*{Completed-load of algorithm $\SLBurst$}


\begin{proofof}{Lemma~\ref{l:typeB}}
Let us fix a pair of arrival and error patterns, such that executions of case (b) occur. Let us now look at the scheduling decisions and performance of each machine individually, after the defined time instant $t$. Note that in such a case, there will only be time intervals of type $T^+$ and/or $T^-$. Otherwise, the execution would be of case (a) since for every time instant $t$ there would exist a future $t'>t$ for which $|L_\lmin(A,t')| < \rhoflr m^2 \bigwedge |L_\lmax(A,t')| < m^2$ would hold.
We define two types of periods for the machine status: the {\em active} and the {\em inactive} periods. During an active period the machine remains alive and the queue of pending tasks does not become empty (recall that the queue of pending tasks never becomes empty in the execution we are studying). An inactive period is a non-active one. In other words, a time interval $[t_r,t_c)$ is active if it starts with time instant $t_r$ such that it is the time right after a restart of the machine. Correspondingly, it ends with time instant $t_c$ such that the machine crashes.
We then focus on the active periods\footnote{We safely ignore the inactive ones since the queue of pending tasks does not become empty and the algorithm $\SLBurst$ is {\em work-conserving}. Hence inactive periods are only while the machine is still crashed.}, with length $\lambda$, which are further categorized in the following four kinds of phases:
\begin{enumerate}
\item Starts with $\lmin$-tasks and has length $\lambda < \rhoflr\lmin$.
\item Starts with $\lmin$-tasks and has length $\lambda \geq \rhoflr\lmin$.
\item Starts with $\lmax$-tasks and has length $\lambda < \lmax$.
\item Starts with $\lmax$-tasks and has length $\lambda \geq \lmax$.
\end{enumerate}

Let as look at the $i^{th}$ period after time $t$ in the execution of $\SLBurst$. Let us also denote by $a_i$ the number of completed $\lmin$-tasks, apart from the $\rhoflr$ preamble, by $b_i$ the number of completed $\lmax$-tasks and by $c_i$ the number of completed $\lmin$-tasks in the preamble. For the execution of $X$ we denote by $a^*_i$ the total number of completed $\lmin$-tasks and by $b^*_i$ the total number of completed $\lmax$-tasks. Let also $C^A(i_j)$ and $C^X(i_j)$ denote the total completed load within a phase $i$ of type $j$ by $\SLBurst$ and $X$ respectively. Analyzing the four types of active periods, we make the following observations.

For phases of type 1, $\SLBurst$ is not able to complete the $\rhoflr$  $\lmin$ tasks of the preamble, while $X$ is only able to complete at most as much load, so $\sum\limits_{\forall i} C^X(i_1) \leq \sum\limits_{\forall i} C^A(i_1)$.

For phases of type 2, the total completed load by $X$ minus the completed load by $\SLBurst$ is at most $\lmax$ (i.e., $\sum\limits_{\forall i} \big( C^X(i_2) - C^A(i_2) \big) < \lmax$). Therefore,
$$
\sum\limits_{\forall i} C^A(i_2) \geq \frac{\rhoflr\lmin}{\lmax + \rhoflr\lmin} \cdot \sum\limits_{\forall i} C^X(i_2).
$$
(Observe that $\frac{\rhoflr\lmin}{\lmax + \rhoflr\lmin} \leq 1/2$.)

The same holds for phases of type 4 
and hence $\sum\limits_{\forall i} C^X(i_4) \leq 2\sum\limits_{\forall i} C^A(i_4)$.

In phases of type 3, $\SLBurst$ is not able complete any task and hence $\sum\limits_{\forall i} C^A(i_3) = 0$, whereas $X$ might complete up to $(\lceil \rho \rceil - 1)\lmin$ tasks. There are two cases of executions to be considered then:\vspace{0.15em}

\indent {\bf Case 1:} The number of phases of type 3 is finite.\vspace{0.15em}

In this case, there is a phase $i^*$ such that $\forall i > i^*$ phase $i$ is not of type 3. Then,
\begin{equation}
\cC_1(A) = \frac{\sum\limits_{j\leq i^*} C^A(j) + \sum\limits_{j>i^*} C^A(j)}{\sum\limits_{j\leq i^*} C^X(j) + \sum\limits_{j>i^*} C^X(j)}
\end{equation}
Observe that the total progress completed by the end of phase $i^*$ by both algorithms is bounded. So for simplicity, we overload notations $A$ and $X$ and define $\sum\limits_{j\leq i^*} C^A(j) = A$ and $\sum\limits_{j\leq i^*} C^X(j) = X$. Therefore,
$$
\cC_1(A) = \frac{A + \sum\limits_{j>i^*} C^A(j)}{X + \sum\limits_{j>i^*} C^X(j)} \geq \frac{A + \frac{\rhoflr\lmin}{\lmax + \rhoflr\lmin}\sum\limits_{j>i^*} C^X(j)}{X + \sum\limits_{j>i^*}C^X(j)}.
$$
Hence, the completed load competitiveness of $\SLBurst$ at the end of each phase can be computed as $\cC(\SLBurst) = \lim_{t\rightarrow\infty}\cC_1(A)$, i.e.,
\begin{eqnarray*}
\cC(\SLBurst) & = & \lim\limits_{j\rightarrow\infty} \frac{A + \frac{\rhoflr\lmin}{\lmax + \rhoflr\lmin}\sum\limits_{j>i^*}C^X(j)}{X + \sum\limits_{j>i^*}C^X(j)} \\
 	& = & \lim\limits_{j\rightarrow\infty} \bigg( \frac{\rhoflr\lmin}{\lmax + \rhoflr\lmin} + \frac{(\lmax + \rhoflr\lmin)A - (\rhoflr\lmin)X}{(\lmax +\rhoflr\lmin)(X + \sum\limits_{j>i^*}C^X(j))} \bigg) \\
 	& = & \frac{\rhoflr\lmin}{\lmax +\rhoflr\lmin} = \frac{\rhoflr}{\rho +\rhoflr}.	
\end{eqnarray*}
It is important to note that the assumption $\lim_{t\rightarrow\infty}C^X(t) = \infty$ is used, which corresponds to the expression $\lim_{j\rightarrow\infty}\sum\limits_{j>i^*}C^X(j)$ in the above equality.

The above analysis shows the completed-load competitiveness at the end of each phase. However, we have to guarantee that the lower bound holds at all times within the phases. For this, consider any time instant $t$ of phase $i>i^*$. At that instant $\cC_i(t) = \frac{\sum_{j\in(i^*,i-1]}C^A(j) + A_t}{\sum_{j\in(i^*,i-1]}C^X(j) + X_t}$, where $A_t$ and $X_t$ represent the load completed by $\SLBurst$ and $X$ within phase $i$ up to time $t$. Using the above proof, and the fact that for phases of type 1,2 and 4 we have
$$
\sum\limits_{\forall i} \big( C^A(i_1) + C^A(i_2) + C^A(i_4) \big)
	\geq  \frac{\lim\rhoflr}{\lmax +\lmin\rhoflr}\cdot \sum\limits_{\forall i} \big( C^X(i_1) + C^X(i_2) + C^X(i_4) \big),
$$
we know that $A_t \geq \frac{\lim\rhoflr}{\lmax +\lmin\rhoflr}\cdot X_t$ as well. Hence,
\begin{eqnarray*}
\cC_i(t) & \geq & \frac{\frac{\rhoflr\lmin}{\lmax + \rhoflr\lmin}\sum_{j\in(i^*,i-1]}C^X(j) + \frac{\rhoflr\lmin}{\lmax+\rhoflr\lmin}X_t}{\sum_{j\in(i^*,i-1]}C^X(j) + X_t} \\
 & = & \frac{\rhoflr\lmin}{\lmax +\rhoflr\lmin} = \frac{\rhoflr}{\rho + \rhoflr}.
\end{eqnarray*}
\vspace{0.15em}

\indent {\bf Case 2:} The number of phases of type 3 is infinite.\\
In this case we must show that the number of $\lmin$ and $\lmax$-tasks completed are bounded for both $\SLBurst$ and $X$.

\begin{claim}
Consider the time instant $t$ at the beginning of a phase $j$ of type 3. Then the number of $\lmin$-tasks completed by $X$ by time $t$ is no more than the number of $\lmin$-tasks completed by $\SLBurst$, plus $\rhoflr -1$, i.e., $\sum_{i<j}a_i^* \leq \sum_{i<j}(a_i+c_i) + (\rhoflr-1)$.
\end{claim}
\begin{proof}
Consider the beginning of phase $j$ of type 3. We know that at that time instant algorithm $\SLBurst$ has at most $(\rhoflr -1)$ $\lmin$-tasks pending. Recall that a machine following algorithm $\SLBurst$, after restarting it first completes a preamble of $\rhoflr\lmin$ tasks, before executing any $\lmax$ ones. By the definition of type 3, it may only occur if there are not enough $\lmin$-tasks pending at time instant $t$. Hence, the amount of $\lmin$-tasks completed by $X$ by the beginning of phase $j$ is no more than the ones completed by algorithm $\SLBurst$ (including the ones in preambles) plus $\rhoflr - 1$.\qed
\end{proof}

\begin{claim}
Considering all types of phases and the number of $\lmax$-tasks completed, it holds that $\sum\limits_{i\leq j}b_i^* \leq \sum\limits_{i\leq j}b_i + \sum\limits_{i\leq j}\frac{c_i}{\rhoflr} + 2$, for every phase $j$.
\end{claim}
\begin{proof}
To prove this, we use induction on phase $j$.\\
\emph{Base Case:} For $j=0$ the claim is trivial.\\
\emph{Induction Hypothesis:} It holds that
$$
\sum\limits_{i\leq j-1}b_i^* \leq \sum\limits_{i\leq j-1}b_i + \sum\limits_{i\leq j-1}\frac{c_i}{\rhoflr} + 2.
$$
\emph{Induction Step:} We need to prove that the relationship holds up to the end of phase $j$. Consider first that during phase $j$ there is a time when $\SLBurst$ has no $\lmax$-tasks pending, and let $t$ be the latest such time in the phase. We define $b^*(t)$ and $b(t)$ being the number of $\lmax$-task completed up to time $t$ by algorithm $X$ and $\SLBurst$ respectively. We know that $b^*(t) \leq b(t)$. We also define $x^*_j(t)$ and $x_j(t)$ to be the number of $\lmax$-tasks scheduled by $X$ and $\SLBurst$ respectively after time instant $t$ and until the end of the phase $j$. We claim that $x^*_j(t) \leq x_j(t) + 2$. From our definitions, at time $t$ algorithm $\SLBurst$ is executing a $\lmin$-task. Since it is the last instant that it has no $\lmax$-task pending, the wort case is to be at the beginning of the preamble (by inspection of the 4 types of phases). Then, if the phase ends at time $t'$, period $I = [t,t']$ is such that $|I| < \rhoflr\lmin + (x_j(t)+1)\lmax \leq (x_j(t)+2)\lmax$. (The +1 $\lmax$-task is because of the machine crash before completing the last $\lmax$-task scheduled in the phase.) Observe that $X$ could be executing a $\lmax$-task at time $t$, completed at some point in $[t,t+\lmax]$ and accounted for in $x_j^*(t)$. Therefore,
$$
\sum\limits_{i\leq j}b^*_j = b^*(t) + x^*_j(t) \leq b(t)+ x_j(t) + 2 = \sum\limits_{i\leq j}b_i + 2.
$$
Now consider the case where at all times of phase $j$ there are $\lmax$-tasks pending for $\SLBurst$. By inspection of the 4 types of phases, the worst case is when $j$ is of type 2. After completing the preamble of $\rhoflr\lmin$ tasks, the algorithm schedules $\lmax$-tasks until the machine crashes again interrupting the last one scheduled. On the same time, $X$ is able to complete at most $\left\lfloor \frac{\lambda_j}{\lmax} \right\rfloor \leq b_j + 1$ $\lmax$-tasks, where $\lambda_j$ is the length of the phase. Hence, in all types of phases $b^*_j \leq \frac{c_j}{\rhoflr} + b_j$ and by induction, the claim follows; $\sum\limits_{i\leq j} b^*_j \leq \sum\limits_{i\leq j} \frac{c_i}{\rhoflr} + \sum\limits_{i\leq j}b_i + 2$.\qed
\end{proof}

Combining the two claims above, the completed load competitiveness ratio of case 2 is as follows:
\begin{align*}
\MoveEqLeft[3] C_2(A) = \frac{\sum\limits_{i\leq j}C^A(i)}{\sum\limits_{i\leq j}C^X(j)} = \frac{\sum\limits_{i\leq j} [(a_i + c_i)\lmin + b_i\lmax]}{\sum\limits_{i\leq j}[a_i^*\lmin + b_i^*\lmax]} \\
	\geq{}& \frac{\sum\limits_{i\leq j} [(a_i + c_i)\lmin + b_i\lmax]}{\sum\limits_{i\leq j}(a_i\! +\! c_i)\lmin\! +\! (\rhoflr\! -\!1)\lmin\! +\! \sum\limits_{i\leq j}(b_i\! +\! \frac{c_i}{\rhoflr})\lmax\! +\! 2\lmax} \\
	\geq{}& \frac{\sum\limits_{i\leq j}[(a_i + c_i)\lmin + b_i\lmax]}{\sum\limits_{i\leq j}[(a_i + 2c_i)\lmin + b_i\lmax] + 3\lmax} \\
	={}& \frac{\sum\limits_{i\leq j}[(a_i + c_i)\lmin + b_i\lmax] + \frac{3}{2}\lmax - \frac{3}{2}\lmax}{2 \sum\limits_{i\leq j}[(a_i+c_i)\lmin + b_i\lmax] + 3\lmax} \\
	\geq{}& \frac{1}{2} - \frac{\frac{3}{2}\lmax}{2 \sum\limits_{i\leq j}[(a_i+c_i)\lmin + b_i\lmax] + 3\lmax}.
\end{align*}

Note that, due to the parameters $a_i,b_i$ and $c_i$, the second ratio tends to zero (the denominator tends to infinity) and hence the completed load competitive ratio tends to $\cC(\SLBurst) = \lim\limits_{t\rightarrow\infty} C_2(A) \geq \frac{1}{2}$.

Combining now the results from the two cases concerning the number of phases of type 3, since $\frac{\rhoflr}{\rho + \rhoflr} \leq \frac{1}{2}$, the completed load of algorithm $\SLBurst$ is at least $\frac{\rhoflr}{\rho + \rhoflr}$ as claimed.\hfill\rule{2mm}{2mm} \!{\sf\tiny Lemma}
\end{proofof}

\subsection*{Completed-load of Algorithm Amortized}


\begin{proofof}{Lemma~\ref{l:l1}}
Let us divide the $n$-busy interval $T = [t,t']$ of $p$ into two intervals; the first being from the beginning, $t$, to a time instant $t^*\geq t$ such that the first restart happens in the interval, and the second being the remaining of the interval, from $t^*$ to $t'$. In other words, interval $T_1 = [t,t^*]$ and $T_2 = [t^*,t']$.

Looking first at interval $T_2$, it starts by a restart and then, either includes more crashes and restarts or not, and never schedules tasks of size more than $\ell_n$. Hence, at time $t^*$ the machine starts executing a new task with both \Amortized and $X$.
Also, since $p$ is busy at all times of the interval, for every $\ell_i$-task completed by $X$ in $T_2$ -- say in interval $T_i = [t_1,t_2]\in T_2$ where $t_2 = t_1 + \ell_i$ and $i \geq n$ -- the machine is able to complete $\rho_{i,n}$ $n$-groups in $T_i$ (each of size $\ell_n$). These groups correspond to executions of the recursive function {\em Schedule\_Group$(n)$}. Hence, we can assign each $n$-group to the task completed by $X$ at the moment when the last task in the $n$-group is completed by \Amortized, which gives inequality $C_p^A(T_2) \geq C_p^X(T_2,\geq\ell_n)$.

Looking now at interval $T_1$, we must consider the following cases for the execution of $p$:\\
(1) At time instant $t$ there was a restart ($t^* = t$) and hence the machine started executing a new task, with both $X$ or \Amortized. In this case the analysis of interval $T_2$ will hold.\\
(2) At time instant $t$ it is already executing a task $\tau$ with $X$, scheduled before $t$ and then, it either a) gets interrupted by the crash at time $t^*$, or b) it completes it within the interval $T_1$. 
In the first case, $X$ is not able to complete any task in $T_1$ while \Amortized may complete up to $|T_1|$, for which it is trivial that $C_p^A(T_1) \geq C_p^X(T_1)$ holds.
In the latter, task $\tau$ will be of maximum size $\ell_k$. Then for the rest of the interval, the same analysis as for $T_2$ holds, for every $\ell_i$-task fully contained in the interval and completed by $X$, where $i \geq n$. Hence, $C_p^A(T_1) \geq C_p^X(T_1, \geq \ell_n) - \ell_k$.

From the two intervals, we have the claim of the lemma, $C_p^A(T) \geq C_p^X(T,\geq\ell_n) - \ell_k$.\qed
\end{proofof}


\begin{proofof}{Lemma~\ref{l:l2}}
The idea of the proof for this lemma, is that tasks completed by algorithm \Amortized are associated to tasks completed by $X$ in such a way, that: (a) each task completed by \Amortized corresponds to at most twice the size of tasks completed by $X$ and (b) each task completed by $X$ is associated to tasks of the same aggregate size completed by \Amortized.
This amortization follows these two rules:\\
\indent 1. the $r^{th}$ task of size $\ell_i$ completed by algorithm \Amortized within $T$, for $i<n$, is associated to the $r^{th}$ task of size $\ell_i$ completed by $X$ within $T$ (if completed).\\
\indent 2. the completion of a task $\tau$ of size $\ell_i \geq \ell_n$ by $X$, corresponds to $\ell_i/\ell_n$ $n$-groups completed by \Amortized, such that the execution of the last task of each of the groups is finished during the execution of task $\tau$.

\vspace{.5em}
First, looking at rule \#1 and interval $[t_1,t]$, the following two equations hold for the pending tasks at the end of the interval:
$$
P^A(t, < \ell_n) = P^A(t_1, < \ell_n) + I([t_1,t], < \ell_n) - C_p^A([t_1,t], < \ell_n),
$$
$$
P^X(t, < \ell_n) = P^X(t_1, < \ell_n) + I([t_1,t], < \ell_n) - C_p^X([t_1,t], < \ell_n),
$$
where $I([t_1,t], < \ell_n)$ is the set of tasks smaller than $\ell_n$ that were injected during the interval, up to time $t$. Since they are the same for both algorithms, from the above equations we have:
$$
P^A(t, < \ell_n) - P^A(t_1, < \ell_n) + C_p^A([t_1,t], < \ell_n)
	= P^X(t, < \ell_n) - P^X(t_1, < \ell_n) + C_p^X([t_1,t], < \ell_n)
$$
which leads to the completed load of \Amortized containing only {\em small} tasks, $< \ell_n$, being bounded as:
\begin{equation}
\label{e:eqA}
C_p^A([t_1,t], < \ell_n) \geq C_p^X([t_1,t], < \ell_n) + P^X(t, < \ell_n) - \sum\limits_{i=1}^{n-1}(\ell_n + \ell_i)(m^2 + m\rho_{i+1}).
\end{equation}
To see why the inequality holds, look first at 
the pseudo-code of the algorithm, more precisely line 6 of algorithm~\ref{alg:schgroup}; a task of size $\ell_n$ is scheduled at time $t$ only in the case when the total size of {\em smaller} pending tasks is $P^A(t,<\ell_n) \leq \sum\limits_{i=1}^{n-1}(\ell_n + \ell_i)(m^2 + m\rho_{i+1})$. 
Recall also condition (3) of the $n$-busy interval of the machine; i.e., $P^A(t_1,\ell_i) \geq P^X(t_1,\ell_i), \forall i\in[1,n]$, which also means that $P^A(t_1, \leq \ell_n) \geq P^X(t_1, \leq \ell_n)$. Combining these properties, the inequality follows.

Now, looking at rule \#2, for any task $\tau$ of size $\ell_i \geq \ell_n$ completed by $X$, we have already shown in Lemma~\ref{l:l1} that
\begin{equation}
\label{e:eqB}
C_p^A\big([t_1,t]\big) \geq C_p^X\big([t_1,t],\geq\ell_n\big) - \ell_k.
\end{equation}

Combining the two equations,~\ref{e:eqA} and~\ref{e:eqB}, the claim follows:
$
2C_p^A\big([t_1,t]\big) \geq C_p^X\big([t_1,t]\big) - \sum\limits_{i=1}^{n-1}(\ell_n + \ell_i)(m^2 + m\rho_{i+1}) - \ell_k.
$\qed
\end{proofof}


\begin{proofof}{Lemma~\ref{l:l3}}
We prove this lemma by induction on $n$.\\
{\em Base case.} For $n = 1$, the result is immediate from Lemma~\ref{l:l1}. More precisely, since $C_p^X(T,\geq\ell_1) = C_p^X(T)$, then $2C_p^A(T) \geq C_p^X(T) - \ell_k$ holds directly.\\
{\em Induction Hypothesis.} We assume that the result holds for some $n<k$, i.e., $2C_p^A(T) \geq C_p^X(T) - f_n$.\\
{\em Inductive Step.} We show that the result still holds for $n+1$. For this, we split the $(n+1)$-busy interval $T$ in three sub-intervals:\vspace{-2mm}
\begin{itemize} \itemsep1pt
\item $T_1$ is the interval from the beginning of $T$ to time instant $t$ at which \Amortized starts executing an $\ell_{n+1}$-task for the last time during $T$.
\item $T_2$ is the interval from $t$ to $t'\in T$ s.t. either \Amortized completes the $\ell_{n+1}$-task, or it gives up scheduling tasks of size $\ell_{n+1}$ at $t'$ since it now has enough smaller tasks pending to cover the $\ell_{n+1}$ time and there was a crash and restart of the machine.
\item $T_3$ from time instant $t'$ to the end of $T$.
\end{itemize}
For sub-interval $T_1$ we know that Lemma~\ref{l:l2} holds, hence
\begin{equation}
\label{e:eq7}
2C_p^A(T_1) \geq C_p^X(T_1) + P^X(t,<\ell_{n+1}) - \sum\limits_{i=1}^n (\ell_{n+1} + \ell_i)(m^2 + m\rho_{i+1}) - \ell_k.
\end{equation} 
Let us now consider an offline algorithm $X'$, which acts as $X$ during $T$, except the fact that it starts sub-interval $T_2$ only with tasks with length at least $\ell_{n+1}$ (no smaller ones), and stays idle whenever $X$ executes a task that was pending in its queue at time instant $t$ but not in the queue of $X'$.

Note here, that algorithm \Amortized finishes the last attempt to complete a $\ell_{n+1}$-task, no longer than $\ell_{n+1}$ time after time instant $t$, in other words, $|T_2| \leq \ell_{n+1}$. Hence, $C_p^{X'}(T_2) \leq \ell_{n+1} + \ell_k$, where $\ell_k$ comes from the possibility of $X'$ scheduling a task before $T_2$ and completing it within $T_2$. Hence,
\begin{equation}
\label{e:eq8}
2C_p^A(T_2) \geq 0 \geq C_p^{X'}(T_2) - \ell_{n+1} - \ell_k.
\end{equation}

At the beginning of sub-interval $T_3$, we have that $P^A(t',\ell_i) \geq P^{X'}(t',\ell_i)$ for each $i\leq n$, since \Amortized only attempted the execution of a $\ell_{n+1}$ during $T_2$ and $X'$ starts the $T_2$ without any tasks smaller than $\ell_{n+1}$. This means that the inductive hypothesis holds for sub-interval $T_3$ for the largest task $\ell_n$ and offline algorithm $X'$ instead of $X$:
\begin{equation}
\label{e:eq9}
2C_p^A(T_3) \geq C_p^{X'}(T_3) - f_n.
\end{equation}

Now observe that at time instant $t$, at the beginning of interval $T_2$, algorithm $X$ has $P^X(t,<\ell_{n+1})$ more tasks pending than $X'$. Hence, by the end of interval $T_3$ the following will hold:
\begin{equation}
\label{e:eq10}
C_p^{X'}(T_2\cup T_3) \geq C_p^X(T_2\cup T_3) - P^X(t,<\ell_{n+1}).
\end{equation}

Putting equations~\ref{e:eq7} to~\ref{e:eq10} together to calculate the completed load of the total interval $T$, we have
\begin{eqnarray*}
2C_p^A(T) & \geq & C_p^X(T_1) + C_p^{X'}(T_2\cup T_3) + P^X(t,<\ell_{n+1}) - \sum\limits_{i=1}^n (\ell_{n+1} + \ell_i)(m^2 + m\rho_{i+1}) - \ell_{n+1} - 2\ell_k + f_n \\
	& \geq & C_p^X(T_1) + C_p^X(T_2\cup T_3) - \sum\limits_{i=1}^n (\ell_{n+1}\! +\! \ell_i)(m^2\! +\! m\rho_{i+1})\! -\! \ell_{n+1}\! -\! 2\ell_k\! -\! f_n \\
	& \geq & C_p^X(T) - \sum\limits_{i=1}^n (\ell_{n+1} + \ell_i)(m^2 + m\rho_{i+1}) - \ell_{n+1} - 2\ell_k - f_n \\
	& \geq & C_p^X(T) - f_{n+1}
\end{eqnarray*}
which completes the induction step and thus the proof of the lemma.\qed
\end{proofof}

\subsection*{Completed-load of Algorithm MAmortized}

The following definition and Claim are necessary for the proof of Lemma~\ref{l:Mamortized}.

\begin{definition}
An execution of {\em Schedule\_Group($k$)} is considered to be {\em uniform}, if the algorithm completes tasks of a single fixed size $\ell_i$ only, during the executions of the current {\em Schedule\_Group($k$)} as well as the previous one.
\end{definition}

\begin{claim}
\label{c:uniform}
There are at least $ck-2k$ uniform calls of {\em Schedule\_Group($k$)} in a stage of M\Amortized.
\end{claim}
\begin{proof}
As already mentioned, the value of $i^*$ can only decrease during a stage. Since there are up to $k$ task sizes, this can happen up to $k-1$ times. However, for a {\em Schedule\_Group($k$)} execution to be uniform, its previous execution must also be uniform. Hence, there can be at least $ck - 2k$ calls of {\em Schedule\_Group($k$)} in a stage that are uniform.\qed
\end{proof}


\begin{proofof}{Lemma~\ref{l:Mamortized}}
Let us start by assuming that all the executions of {\em Schedule\_Group} in M\Amortized are uniform.
Every task $\ell_i \geq \ell_n$ completed by $X$, will correspond to a group of $\lfloor \ell_i/\ell_n \rfloor$ tasks of total size $\ell_i$ completed by algorithm M\Amortized. In particular, to the group of tasks that completed their execution successfully during the execution of the task $\ell_i$ by $X$.

Now let $\gamma = \min\limits_{i\leq j < i \leq k} \Big\{ \frac{\overline{\rho_{i,j} }}{\overline{ \rho_{i,j} } + \rho_{i,j}} \Big\}$. Let us also define $\delta_{i,j} = \frac{\rho_{i,j}}{\overline{\rho_{i,j} }}$ and $\delta = \max\limits_{i>j} \{\delta_{i,j} \}$. This means that $\gamma = 1/(1+\delta)$. The assignment of tasks completed by M\Amortized, to each $\ell_i$-task completed by $X$, where $i\geq n$, makes the inequalities from Lemmas~\ref{l:l1},~\ref{l:l2} and~\ref{l:l3} as follows:
$$
\delta\cdot C^A(T) \geq C^X(T, \geq \ell_n) - \ell_k
$$
$$
(1 + \delta)\cdot C^A([t_1,t]) \geq C^X([t_1,t]) + P^X(t,<\ell_n) - \sum\limits_{i=1}^n (\ell_{n+1} +\ell_i)(m^2 + m\rho_{i+1}) - \ell_n
$$
$$
(1 + \delta)\cdot C^A(T) \geq C^X(T) - f_n
$$
Then, if we apply the above inequalities in the proof of Theorem~\ref{th:Amortized}, we obtain the claimed result, without the $c'$ factor, since $\frac{1}{1+\delta} = \min\limits_{1\leq j<i\leq k} \Big\{ \frac{\overline{\rho_{i,j} }}{\overline{\rho_{i,j}} + \rho_{i,j}} \Big\}$.

Nonetheless, the above case only covers the uniform executions. Let us now consider the cases where there is some execution of {\em Schedule\_Group} that is not uniform. By Claim~\ref{c:uniform}, at most a fraction of calls to {\em Schedule\_Group($k$)} are not uniform; that is, $2/c$ of them. From Claim~\ref{c:mod1}, even without the uniform executions we have the completed load competitiveness at least $\eta = \min\limits_{i\in [2,k]} \Big\{ \frac{\rho_i - 1}{2\rho_i - 1} \Big\}$.

Combining the two cases, we separate the uniform and the non-uniform executions, and denote the corresponding completed load of $X$ by $C^X_1(T)$ and $C^X_2(T)$ respectively. We therefore have the following relationships between the two algorithms:
$$
\left(1 - \frac{2}{c}\right) \cdot C^A(T) \geq \gamma \cdot C^X_1(T) - c_1
$$
$$
\frac{2}{c} \cdot C^A(T) \geq \eta \cdot C^X_2(T) - c_2
$$
for constants $c_1$ and $c_2$ that depend on the task sizes. This means that
$$
\left( 1 - \frac{2}{c} + \frac{2\gamma}{c\eta} \right) \cdot C^A(T) \geq \gamma \cdot \Big( C^X_1(T) + C^X_2(T) \Big) - c_1 - c_2
$$
which leads to the desired completed load competitiveness:
$$
C^A(T) \geq \frac{1}{1 + 2/(c\eta)}\cdot \gamma\cdot C^X(T) - c_1 - c_2
$$
where we can define the $c'$ of the lemma equal to $\frac{1}{1 + 2/(c\eta)}$ and choose such a $c$, large enough, to make $c'$ arbitrarily close to 1. This completed the proof of the lemma.\qed
\end{proofof}

\subsection*{Completed-load in the case of speedup}

\begin{proofof}{Theorem~\ref{t:srhoplus}}
We consider any distributed work-conserving algorithm $\Alg_W$, running on $m$ parallel machines with speedup $s \geq 1 + \rho$.
For the proof of the theorem let us consider only the periods of execution during which there are at least $m^2$ pending tasks. 
During the remaining time of the executions, the completed load is bounded by the number of tasks pending (i.e., $< m^2\lmax$).


Consider then, the execution of machine $p$ and the corresponding execution of offline algorithm $X$ in the same machine. We will be looking at their completed load by machine $p$ at different time instances. Let us look at any time $t$ and define time instant $t'<t$ to be the latest time before $t$ at which one of the following two events happens: (1) and {\em active} period starts ($t'$ is a restart point of $p$), or (2) algorithm $X$ has successfully completed a task.

It is then trivial that $C^p_0(\Alg,A,E) \geq C^p_0(X,A,E)$ holds at the beginning of the executions. Now, assuming that $C^p_{t'}(\Alg,A,E) \geq C^p_{t'}(X,A,E)$ holds at time $t'$, we prove by induction that $C^p_t(\Alg,A,E) \geq C^p_t(X,A,E)$ still holds at time $t$. This also means that the tasks successfully completed by machine $p$ in the execution of \Alg by time $t$ have at least the same total size as the ones completed by $X$.

Looking at the interval $T = (t',t]$, we have to consider the following two cases:

\noindent {\em Case 1: $X$ is not able to complete any task in the interval $T$.} This means that $C^p_t(X,A,E) = C^p_{t'}(X,A,E)$. For \Alg it holds that $C^p_t(\Alg,A,E) \geq C^p_{t'}(\Alg,A,E)$, even if it is not able to complete any tasks in $T$. Therefore, $C^p_t(\Alg,A,E) \geq C^p_t(X,A,E)$.

\noindent {\em Case 2: $X$ completes a task in the interval $T$.} Note that, due to the definition of $t'$, there can only be one task completed by $X$ within $T$, and it must be completed exactly at time instant $t$. Then, interval $T$ has length equal to a task $\ell \in [\lmin,\lmax]$; the size of the task completed by $X$. In $T$ algorithm \Alg executes tasks continuously, whose aggregate size is at least $\ell s - \lmax$. Then, the completed load of machine $p$ with the two algorithms at $t$ satisfies $C^p_t(X,A,E) = C^p_{t'}(X,A,E) + \ell$ and $C^p_t(\Alg,A,E) \geq C^p_{t'}(\Alg,A,E) + (\ell s - \lmax)$. Observe that the fact that $s\geq 1 + \rho$ implies that $\ell s - \lmax \geq \ell$. Hence, $C^p_t(\Alg,A,E) \geq C^p_t(X,A,E)$.

The above analysis shows that the completed-load competitiveness at any time of the execution of each machine is 1. Therefore, taking the sum of completed loads of all machines, gives the claimed result; the completed-load competitiveness ratio of any work-conserving distributed algorithm that guarantees non redundant executions while there are at least $m^2$ tasks pending, is $\cC(\Alg) \geq 1$, when $s \geq 1 + \rho$. \qed
\end{proofof}

Let us present here the pseudo-code of algorithm $\LISs$.

\lstset{columns=fullflexible,
	tabsize=3,
	columns=flexible,
	mathescape=true,
	morekeywords={Parameters, Upon, awaking, or, restart, Repeat, Get, Sort, If, then, Inform, Set, else},
}
\lstset{escapeinside={('}{')}}

\setlength{\textfloatsep}{2pt}
\begin{algorithm}[H]
\caption{ {\bf $\LISs$} (for machine $p$)}
\label{algLIS}
\begin{lstlisting}
Parameters: $m, \lmin, \lmax$
Upon awaking or restart
	Repeat
		Get sorted queue $Q$ from the Repository;
		If $|Q| \geq m^2$ then
			Schedule task $\ell$ at position $p \cdot m$ in $Q$;
		else
			Schedule task $\ell$ at position $(p \cdot m) \mod |Q|$;
		Inform Repository of completion of task $\ell$;
\end{lstlisting}
\end{algorithm}

\begin{proofof}{Theorem~\ref{t:mLIS} ({\em Sketch})}
We look at the case when $\LISs$ runs in a parallel system of two machines ($m=2$) and speedup $s = \rho = 2$. Let us fix an adversarial strategy, consisting of task arrival and machine error patterns $A$ and $E$, that work as follows:

We define $\delta=\rho^{1/5} \approx 1.15$, and use only tasks of sizes $x \cdot \lmin$, for $x \in \{1, \delta, \delta^2, \delta^3, \delta^4, \delta^5\}$. For simplicity,
in the rest of the proof we remove the factor $\lmin$, which only introduces a scaling factor.

The arrival pattern $A$ is the following sequence of task sizes that is repeated over and over:
\begin{equation}
\label{lis-arrival}
1, \delta^2, \delta^2, \delta^4, \delta, \delta^3, \delta^3, \delta^5, \cdots
\end{equation}
Let the arrival of tasks be fast enough so that whenever algorithms $\LISs$ or $X$ are supposed to schedule a task in the description below such a task is in the repository.\\

The execution then behaves as follows: We divide the execution in {\em epochs}, so that in one epoch $\LISs$ executes the tasks that are in positions 1 to 8 in the repository sorted by arrival time, as shown above. As defined by $\LISs$, processor 1 always schedules the task in position 1, while processor 2 schedule the task in position 3. In each epoch there are 8 {\em phases} as follows:
\begin{enumerate}
\item 
In the first phase, processor 1 is crashed, while processor 2 restarts, is active for time $\delta^3$, and the crashes again. In this phase,
$X$ schedules and completes a task of length $\delta^3$. On its hand, $\LISs$ schedules and completes a task of length $\delta^2$, and schedules
a task of length $\delta^4$ that is interrupted, since $\frac{\delta^2+\delta^4}{2} >\delta^3$.
\item
In the second phase, processor 2 stays crashed. Processor 1 is active for time $\delta$. In this phase $X$ schedules and completes a task of length $\delta$.
$\LISs$ schedules and completes a task of length $1$, and schedules a task of length $\delta^2$ that is interrupted.
\item
In the third phase, processor 2 remains crashed. The tasks in positions 1 and 2 at the start of the task have lengths $\delta^2$ and $\delta^4$. In this phase,
processor 1 is active for $\delta^3$ time. As in the first phase, $X$ completes a task of length $\delta^3$, while $\LISs$ only completes a task of length $\delta^2$.
\item
In the fourth phase, processor 2 keeps being crashed. The tasks in positions 1 and 2 at the start of the task have lengths $\delta^4$ and $\delta$.
In this phase, processor 1 is active for $\delta^2$ time. Hence, $X$ completes a task of length $\delta^2$ while $\LISs$ only completes a task of length $\delta^4$.
\end{enumerate}
These first four phases complete the execution of the 4 first tasks in the sequence \ref{lis-arrival} above. The next three phases of the epoch are
similar to phases 1 to 3, but all task lengths have an additional factor $\delta$. In the final eighth phase, only processor 1 is active, for $\delta^2$ time,
$X$ completes a task of length $\delta^2$, and $\LISs$ completes a task of length $\delta^5$ (the next task in arrival order is the task of length 1 that starts
a new sequence of 8 tasks like sequence \ref{lis-arrival} above).

The total length of the tasks completed by $X$ in the first epoch is then $C(X)=2\delta^4+2\delta^3 + 3\delta^2+\delta \approx 11.62$, while $n\LIS$ has completed
$C(\LISs)=\delta^5+\delta^4+2\delta^3+2\delta^2+\delta+1 \approx 11.56$. Exactly the same behavior is repeated in every epoch.
Hence, $m\LIS$ is not 1-competitive with $s=\rho=m=2$.\qed
\end{proofof}

\end{document}